\documentclass[11pt]{article}
\usepackage[T1]{fontenc}
\usepackage[margin=1in]{geometry}
\usepackage{mathtools,amssymb,amsthm,bm}
\usepackage{mathpazo,microtype,enumitem,xcolor}

\definecolor{referenceRed}{rgb}{0.58,0,0}
\definecolor{citationGreen}{rgb}{0,0.31,0.11}
\usepackage[colorlinks=true,linkcolor=referenceRed,citecolor=citationGreen,urlcolor=blue,pagebackref]{hyperref}
\usepackage[nameinlink,noabbrev]{cleveref}

\newtheorem{theorem}{Theorem}[section]
\newtheorem{lemma}[theorem]{Lemma}
\newtheorem{proposition}[theorem]{Proposition}
\newtheorem{corollary}[theorem]{Corollary}
\theoremstyle{definition}
\newtheorem{definition}[theorem]{Definition}
\newtheorem{innercustomlem}{Lemma}
\newenvironment{mylem}[1]{\renewcommand\theinnercustomlem{#1}\innercustomlem\itshape}
  {\endinnercustomlem}

\newcommand{\EE}{\mathbb E}
\newcommand{\PP}{\mathbb P}
\newcommand{\paren}[1]{\left(#1\right)}
\newcommand{\ind}{\mathbf 1}
\DeclareMathOperator{\poly}{poly}
\DeclareMathOperator{\Tr}{Tr}

\setlist{itemsep=3pt,topsep=5pt}
\allowdisplaybreaks[2]

\title{Color Coding for the Sherrington--Kirkpatrick Model}
\author{\begin{tabular}{c@{\hspace{3em}}c}
Alina Harbuzova & Saba Lepsveridze\\
{\small MIT, \texttt{hadought@mit.edu}} &
{\small MIT, \texttt{sabal@mit.edu}}\\[0.6em]
Mahbod Majid\thanks{Supported by the 2026 Apple Scholars in AI/ML PhD fellowship} & Ankur Moitra\\
{\small MIT, \texttt{mahbod@mit.edu}} &
{\small MIT, \texttt{moitra@mit.edu}}
\end{tabular}}
\date{}

\begin{document}
\backrefsetup{disable}
\pagenumbering{gobble}
\pagestyle{empty}
\maketitle
\thispagestyle{empty}
\begin{abstract}
We give a polynomial-time algorithm for approximating the partition function of mean-field mixed $p$-spin models to arbitrarily high accuracy throughout the second-moment regime. This improves the quasipolynomial running time of Bencs et al.~\cite{BHL25} and answers their open question. In particular, our result covers the entire replica-symmetric regime of the Sherrington--Kirkpatrick model. The algorithm is deterministic, runs in time polynomial in $n$ and $1/\varepsilon$, and succeeds for every typical realization of the disorder.

Our main technical contribution is a new application of color coding to the classical high-temperature expansion developed by Aizenman, Lebowitz, and Ruelle~\cite{ALR87} to study fluctuations of the partition function. We combine this approach with the zero-freeness from ~\cite{BHL25} to convert additive approximations into multiplicative ones.

\end{abstract}
\clearpage
\begingroup
\setcounter{tocdepth}{2}
\fontsize{10}{10.8}\selectfont
\tableofcontents
\endgroup
\clearpage
\pagestyle{plain}
\pagenumbering{arabic}
\hypersetup{pageanchor=true}
\backrefsetup{enable}

\section{Introduction}\label{sec:introduction}

Mean-field spin glasses are random probability distributions introduced half a century ago by Sherrington and Kirkpatrick~\cite{SK75} as idealized mathematical models for disordered magnetic materials. Since then, they have become fundamental objects of study at the intersection of statistical physics, high-dimensional inference, and average-case complexity. 

A central computational question in this setting is whether the partition function of a typical instance can be approximated efficiently to a prescribed accuracy. To be more precise, we consider \emph{mixed $p$-spin Ising models}, whose random \emph{Hamiltonians} are defined on  $\sigma \in \{\pm 1\}^n$ by

\begin{equation}\label{eq:model}
 H_G(\sigma)=\sum_{p=2}^{p_{\max}}
 \frac{\gamma_p}{n^{(p-1)/2}}
 \sum_{i_1,\ldots,i_p=1}^n
 G_{i_1,\ldots,i_p}\sigma_{i_1}\cdots\sigma_{i_p}.
\end{equation}

Here, $p_{\max} \ge 2$ is a fixed integer, $\gamma_2,\ldots,\gamma_{p_{\max}}$ are fixed real coefficients, and the entries $G_{i_1,\ldots,i_p}$ are independent standard Gaussians across all indices. Note that the original Sherrington--Kirkpatrick model corresponds to the special case $p_{\max}=2$ and $\gamma_2=1/\sqrt{2}$.

Given a Hamiltonian $H_G$ and a fixed inverse temperature $\beta \geq 0$, the corresponding \emph{Gibbs measure} on $\{\pm 1\}^n$ is defined by

\begin{equation*}
    d \mu_{G,\beta}(\sigma) \propto \exp\{\beta H_G(\sigma)\} \times d\rho(\sigma),
\end{equation*}
where normalization constant is the \emph{partition function} is
\begin{equation*}
    Z_G(\beta) = \mathbb{E}_{\sigma \sim \rho}\exp\{\beta H_G(\sigma)\}.
\end{equation*}
Here $\rho$ distributed uniformly on $\{\pm 1\}^n$. Given $G$, $\beta$, and $\varepsilon > 0$, our goal is to efficiently compute an estimate $\widehat Z_G(\beta)$ satisfying
\[
    \left|\frac{\widehat Z_G(\beta)}{Z_G(\beta)}-1\right| \le \varepsilon
\]
with high probability over $G$. We are also  interested in the range of  $\beta$ for which this task can be performed
efficiently.

\medskip

Much of the theory of spin glasses concerns their physical
properties, such as the limiting free energy and replica symmetry
breaking~\cite{Tal06,Pan13}. Over time, important connections
have emerged between the computational  and
physical phase transitions in these systems.

For example, for the hard-core model on bounded-degree graphs,
approximate counting is efficient below the uniqueness threshold
and hard in the worst case above it~\cite{Weitz06,SS14}.
Moreover, for the Sherrington--Kirkpatrick model, polynomial-time
sampling is possible for every fixed $\beta<1$, with vanishing error
in both normalized Wasserstein distance~\cite{AMS22,Cel22} and
total variation distance~\cite{LSS26}. For $\beta>1$, however,
disorder chaos rules out stable samplers even in normalized
Wasserstein distance~\cite{AMS22}.

For mixed $p$-spin models, a natural threshold arises from
the second moment of the partition function. To make this precise, define the mixture function
\[
    \xi(q)=\sum_{p=2}^{p_{\max}}\gamma_p^2q^p.
\]
Since  $\mathbb{E}_G[H_G(\sigma)H_G(\tau)]
    =n\xi({\langle\sigma,\tau\rangle}/{n})$, a Gaussian moment calculation and a large-deviation estimate for the overlap gives that
\begin{equation}\label{eq:partition-second-moment}
    \frac{\mathbb{E}_G[Z_G(\beta)^2]}
         {(\mathbb{E}_G Z_G(\beta))^2}=
    \mathbb{E}_{\sigma,\tau\sim\rho}
    \exp\!\left\{
        n\beta^2\xi\!\left(\frac{\langle\sigma,\tau\rangle}{n}\right)
    \right\} =
    \exp\!\left\{
        n\sup_{0\leq q\leq1}
        \bigl[\beta^2\xi(q)-I(q)\bigr]+o(n)
    \right\},
\end{equation}
where $\sigma$ and $\tau$ are independent and
\[
    I(q)=\frac{(1+q)\log(1+q)+(1-q)\log(1-q)}{2}
\]
is the entropy. This leads to the second-moment threshold
\begin{equation}\label{eq:second-moment-threshold}
    \beta_{\mathrm{2nd}}
    :=
    \left(\inf_{0 < q \leq 1}\frac{I(q)}{\xi(q)}\right)^{1/2}.
\end{equation}
For every fixed
$\beta>\beta_{\mathrm{2nd}}$, the second-moment ratio grows exponentially in $n$, while for 
$\beta<\beta_{\mathrm{2nd}}$, it remains bounded as $n\to\infty$~\cite{BHL25}. We call this range the \emph{second-moment regime}.

\subsection{Our Results}
\label{sec:our-results}

Our main result gives a deterministic polynomial-time algorithm for approximating the partition function of mixed $p$-spin Ising models throughout this second-moment regime,  with high probability over the disorder. 

\begin{theorem}
\label{thm:counting-intro}
Fix $\xi(q)=\sum_{p=2}^{p_{\max}}\gamma_p^2q^p$ and
$0<\beta_*<\beta_{\mathrm{2nd}}$. There exist a constant $C$,
a sequence\footnote{Here $\zeta_n$ bounds the probability that $Z_G$ has a complex zero in the disk at the origin with radius $\beta_*$, as given by the zero-freeness theorem of Bencs et al.~\cite{BHL25}.} $\zeta_n=o(1)$, and a deterministic algorithm with
the following guarantee.

The algorithm takes as input the disorder $G$ and parameters
$\varepsilon,\delta>0$, and outputs an estimator
$\widehat Z_G(\beta)$ for any given real $|\beta| \leq \beta_*$ satisfying
\[
\mathbb P_G\!\left[
\sup_{|\beta|\leq\beta_*}
\left|\frac{\widehat Z_G(\beta)}{Z_G(\beta)}-1\right|
\leq\varepsilon
\right]\geq1-\zeta_n-\delta.
\]
Evaluating $\widehat Z_G(\beta)$ requires at most
$(n/(\varepsilon\delta))^C$ operations in the exact-arithmetic
model.\footnote{Arithmetic operations, comparisons, and evaluations
of elementary functions have unit cost.}
\end{theorem}

In particular, for the SK model, $\beta_{\mathrm{2nd}}=1$, so our
result gives polynomial-time counting throughout the
high-temperature phase.

Our work crucially builds on the result of Bencs, Huang, Lee, Liu, and
Regts~\cite{BHL25}, who obtained a deterministic quasipolynomial-time algorithm for counting in the same regime. For every fixed radius $0 < \beta_* < \beta_{\mathrm{2nd}}$, they prove that there exists $\zeta_n = o(1)$ such that the partition function $Z_G(z)$ is nonzero in the complex disk $|z| \leq \beta_*$, with probability at least $1-\zeta_n$. Using this zero-freeness and Barvinok's interpolation method~\cite{Bar16book}, they approximate $Z_G$
by truncating the Taylor expansion of $\log Z_g$ at degree $C\log(n/\varepsilon)$.
Brute force computation of the Taylor coefficients takes $n^{C\log(n/\varepsilon)}$ time. \Cref{thm:counting-intro} improves this running time to polynomial.

Our main contribution is to reduce the runtime of this algorithm. The natural idea is to use color coding~\cite{AYZ95}, which can efficiently evaluate sums indexed by sparse graphs via dynamic programming without enumerating their terms individually. The difficulty is that, even for the SK model, the Taylor expansion used in~\cite{BHL25} contains terms indexed by multigraphs with high-degree vertices, making it unclear how to apply color coding directly.

We instead apply color coding to the high-temperature expansion of Aizenman, Lebowitz, and Ruelle~\cite{ALR87}, originally introduced to study fluctuations of the SK partition function. For SK, this expansion writes the partition function as a weighted sum over even subgraphs of the complete graph. Aizenman, Lebowitz, and Ruelle showed that unions of vertex-disjoint cycles give the leading contribution to the sum. While this suffices to study the limiting fluctuations, to obtain a prescribed accuracy, we also retain graphs with few edges and few branching vertices, namely vertices of degree at least four. Once the branching vertices are fixed, every retained graph decomposes into paths joining them and cycles disjoint from them. Color coding and dynamic programming then allow us to evaluate the total weight of these graphs efficiently. We explain the argument in more detail in \Cref{sec:overview}.

For mixed $p$-spin models, we additionally control the number of nonpair interactions. After fixing their incident vertices together with the branching vertices, the remaining computation again reduces to paths and cycles. The full argument is in \Cref{app:mixed}. 

\subsection{Acknowledgements}

The authors explored several ways to apply color coding to the
polynomial used by Bencs et al.~\cite{BHL25} to improve their
quasipolynomial running time to polynomial. In the course of this
exploration, OpenAI's GPT-6 Astra identified the high-temperature
graph expansion as a more convenient object for color coding.
The central argument based on this choice is due to GPT-6 Astra.
AI assistance was also used to develop arguments and to
prepare the manuscript. The authors take responsibility for the 
mathematical claims, proofs, and citations.

Alina Harbuzova is supported by Schwarzman College of Computing PhD Fellowship.

\paragraph{Concurrent work.}
Schmidhuber and Zlokapa~\cite{SZ26} concurrently and independently obtained deterministic polynomial-time counting algorithms for SK throughout \(\beta<1\), using closely related graph-expansion and color-coding techniques. For mixtures without any quadratic interactions, however, they use different approach based on cavity method that reaches the static transition threshold $\beta_{\mathrm {stat}} \geq \beta_{\mathrm{2nd}}$. Our result extends color coding argument to general finite mixtures, including quadratic and higher-order interactions together.

The recent OpenAI mathematics release also includes
manuscripts on the spectral gap and cutoff of SK Glauber
dynamics~\cite{OAI:A-spectral-gap-throughout-the-high-temperature-Sherrington-Kirkpatrick-phase-September-24-2026,OAI:Cutoff-throughout-the-high-temperature-Sherrington-Kirkpatrick-phase-September-24-2026},
obtaining $O_\beta(n\log n)$ worst-start mixing at fixed total-variation
accuracy for every fixed $0\le\beta<1$, with high probability over the
disorder.
Their proofs use stochastic localization and estimates for posterior Gibbs measures. The spectral-gap argument also yields randomized polynomial-time approximate counting by making its estimates uniform on compact subcritical temperature intervals and applying standard annealing. Thus these results provide a sampling-based route to counting, whereas our algorithm is deterministic and directly evaluates the high-temperature expansion using color coding. Conversely, our counting theorem does not establish a spectral gap or rapid mixing of Glauber dynamics.

\section{Technical overview}\label{sec:overview}

In this section, we explain the argument behind Theorem~\ref{thm:counting-intro}. For clarity, we present the argument for the SK model; the modifications required for mixed \(p\)-spin models are described in \Cref{app:mixed}.

We fix a target interval $|\beta|\le\beta_*$ with $0<\beta_*<1$.
Let $\EE_\sigma$ denote the uniform average over $\{\pm 1\}^n$ and recall that our goal is to estimate the partition function $$Z_g(\beta) = \EE_{\sigma}e^{\beta H_g(\sigma)}\,,$$ where $H_g(\sigma)=n^{-1/2}\sum_{i<j}g_{ij}\sigma_i\sigma_j$
and $(g_{ij})_{i<j}$ are independent standard Gaussians. Our algorithm approximates instead the \emph{normalized partition function} $$Y_g(\beta) = \frac {Z_g(\beta)}{\EE_g Z_g(\beta)}\,.$$

Since the normalization factor $\EE_g Z_g(\beta) = e^{(n-1)\beta^2/4}$ is explicit and easily computable, an estimate for $Y_g(\beta)$ gives the estimate for $Z_g(\beta)$ with the same multiplicative approximation error. 

\bigskip

We first express the normalized partition function as an explicit prefactor times a sum over graphs, using the high-temperature expansion of Aizenman, Lebowitz, and Ruelle~\cite{ALR87,ALR88}. Denote the set of edges of a complete graph on $\{1, \ldots, n\}$ by $E$. We can view the Hamiltonian $H_g(\sigma)$ as a sum over the edges $\{i, j\} \in E$. For each edge $\{i, j\}$, 
\[
 e^{\beta g_{ij}\sigma_i\sigma_j / \sqrt{n}}
 =\cosh(\beta g_{ij}/\sqrt{n})\paren{1 + \sigma_i \sigma_j\tanh (\beta g_{ij}/\sqrt{n}) }\,.
\]
Then, denoting $t_{ij}(\beta)\coloneqq\tanh(\beta g_{ij} / \sqrt{n})$, we can express
\begin{equation}\label{eq:prefactor}
Y_g(\beta) = \frac{\prod_{\{i,j\}\in E}\cosh(\beta g_{ij} / \sqrt{n})}{\EE_g Z_g(\beta)} \cdot \underbrace{\EE_\sigma\prod_{\{i,j\}\in E}(1+t_{ij}(\beta)\sigma_i\sigma_j)}_{\eqcolon P(t)}\,.
\end{equation}
The prefactor is positive and computable in polynomial time at each real temperature, so this reduces the task of approximating $Y_g(\beta)$ to approximating $P(t(\beta))$. For a subset of edges $F \subseteq E$, write $t_F\coloneqq\prod_{\{i,j\}\in F}t_{ij}$, with $t_\varnothing\coloneqq1$. Then, 
\begin{equation}
 P(t) =\EE_\sigma\prod_{\{i,j\}\in E}(1+t_{ij}\sigma_i\sigma_j) =\sum_{F\subseteq E}t_F\prod_{i=1}^n\EE_\sigma\sigma_i^{\deg_F(i)}
 =\sum_{F\subseteq E \text{ even}}t_F\,.\label{eq:ov-graph-expansion}
\end{equation}
Here a subgraph \(F\) is called \emph{even} if every vertex has even degree in \(F\). 

\medskip

The main novelty is to construct an efficiently computable deterministic estimate $\widehat P$ satisfying 
\begin{equation}\label{eq:overview-estimator-additive}
    \sup_{|\beta| \leq \beta_*}\EE_g\left|\widehat P(t(\beta))-P(t(\beta))\right|^2 \leq n^{-D}
\end{equation}
for any given constant $D>0$. Before describing $\widehat P$, we explain how the coefficient estimate underlying the additive estimate in \eqref{eq:overview-estimator-additive} lead to the desired multiplicative guarantee in \Cref{thm:counting-intro}.

In \Cref{sec:analysis}, we convert the estimate leading to \eqref{eq:overview-estimator-additive} into a uniform error guarantee on $\widehat Y_g$, which is obtained by multiplying $\widehat P$ by the explicit prefactor in \eqref{eq:prefactor}:
\begin{equation}\label{eq:overview-estimator-additive2}
    \EE_g \sup_{|\beta| \leq \beta_*}\left|\widehat Y_g(\beta)-Y_g(\beta)\right|^2 \leq n^{-D}\,.
\end{equation}
Finally, to obtain the multiplicative guarantee $\bigl| \frac{\widehat Y_g(\beta)}{Y_g(\beta)} -1 \bigr|\leq \varepsilon$, we use Markov's inequality to bound
\begin{equation*}
    \PP_g\Big[ \sup_{|\beta| \leq \beta_*}\big|\widehat Y_g-Y_g\big| \geq \varepsilon \inf_{|\beta| \leq \beta_*}\big|Y_g\big| \Big]
    \leq \PP_g\Big[ \inf_{|\beta| \leq \beta_*}\big|Y_g\big| \leq \gamma \Big]
    + \frac{\EE_g \sup_{|\beta| \leq \beta_*} \big|\widehat Y_g-Y_g\big|^2}{\varepsilon^2\gamma^2}
\end{equation*}
for every $\varepsilon,\gamma>0$. We then use the zero-freeness theorem of \cite{BHL25} to bound $\PP_g\bigl[ \inf_{|\beta| \leq \beta_*}|Y_g| \leq \gamma \bigr]$, thereby reducing our task to bounding $\sup_{|\beta| \leq \beta_*}\EE_g\big|\widehat P(t(\beta))-P(t(\beta))\big|^2$.

\medskip

We now construct an efficiently computable estimate $\widehat P$. \cite{ALR87} prove that the expansion of $P$ in \eqref{eq:ov-graph-expansion} is dominated by unions of disjoint cycles. To achieve a prescribed accuracy guarantee, we refine this approximation by including graphs $F$ with a small number of \emph{branching vertices}, that is, vertices of degree at least \(4\). 

Note that the monomials corresponding to distinct graphs $F \neq F'$ are orthogonal under the Gaussian disorder due to symmetry :  $\EE_g [t_F t_{F'}] = 0$. Consequently, for any family of graphs $\mathcal F$, we have 
$$
\EE_g \Bigl| \sum_{F \in \mathcal F} t_F \Bigr|^2 = \sum_{F \in \mathcal F} \EE_g[t_F^2] = \sum_{F \in \mathcal F} \prod_{\{i,j\} \in F }\EE_g \tanh^2(\beta g_{ij} /\sqrt{n}) \leq \sum_{F \in \mathcal F} \left(\frac {\beta^2} n \right)^{|F|}\,,
$$
Now, for a fixed even graph $H$ with $v$ vertices and $e$ edges, let $\mathcal F_H$ denote its copies on $[n]$. Then,
$$
\EE_g \Bigl| \sum_{F \in \mathcal F_H} t_F \Bigr|^2 \leq n^v \left(\frac {\beta^2} n \right)^{e} = \beta^{2e} n^{v-e}\,.
$$
Note that every branching vertex contributes to $e-v$ as
$$
e - v = \frac 12 \sum_{i: \deg_H(i)>0} (\deg_H(i)-2) \geq |\{ i: \deg_H(i) \geq 4 \}|\,.
$$
Thus $e=v$ for precisely disjoint union of cycles, while every other graph has $e > v$, that introduces an additional factor $n^{v-e}$ in the second-moment bound. For our approximation, we retain graphs with at most $K$ edges and fewer than $b$ branching vertices. Thus, we define
\[
 P_{K,b}(t)\coloneqq\sum_{F\in\mathcal F_{K,b}}t_F\,, \qquad \text{where }\ \mathcal F_{K,b}\coloneqq\{F\subseteq E:F\text{ even},\ |F|\le K,\ |\{i: \deg_F(i)\geq 4\}|<b\}\,.
\]
In \Cref{sec:truncation}, we bound the expected squared error \(\mathbb E_g|P(t)-P_{K,b}(t)|^2\) uniformly over real \(|\beta|\le\beta_*\), as stated in \Cref{lem:sk-truncation}. In particular, for every fixed \(D>0\), we can make this error at most \(n^{-D}\) by choosing suitable cutoffs
\[
K=O_D(\log n)
\qquad\text{and}\qquad
b=O_D(1).
\]

Note that if we could efficiently compute $P_{K,b}(t)$, we would be done. However, naive enumeration of $\mathcal F_{K,b}$ takes quasipolynomial time, since it contains as many graphs. Instead, we construct an efficiently computable deterministic estimate $\widehat P$ of $P_{K,b}$.

\medskip

To this end, we use the color coding method introduced by Alon, Yuster, and Zwick~\cite{AYZ95}. Let $h:[n]\to[K]$ be a uniformly random balanced coloring, and call a graph \emph{colorful} if its vertices have distinct colors. Denoting by $p_v$ the probability that a fixed set of $v$ vertices is colorful, we construct the following estimator:
\begin{equation}\label{eq:color-estimator-sketch-ov}
V_h\coloneqq\sum_{F\in\mathcal F_{K, b}}
\frac{\ind\{F\text{ colorful under }h\}}
     {p_{|V(F)|}}t_F.
\end{equation}
To obtain a more accurate estimate, we sample $T$ balanced colorings $h_1,\dots,h_T$ independently and set
$$
\widehat P \coloneq \frac1T\sum_{t=1}^T V_{h_t}\,.
$$
In \Cref{sec:estimator}, we show that
$$\EE_{h_1,\dots,h_T}\EE_g|\widehat P-P_{K,b}|^2\leq e^{O(K)}/T \leq n^{-D}$$
for any prescribed $D >0$, where the last inequality follows from taking a sufficiently large $T=\poly(n)$ in the regime $K=O(\log n)$.

To obtain a deterministic estimate, we replace the random colorings by deterministic ones, using the method of conditional expectations~\cite{ES73,Rag88}, see \Cref{subsec:color-derandomization}.

\bigskip

It remains to show that the estimate $\widehat P$ can be computed efficiently in polynomial time. In particular, it suffices to show that for any given coloring $h$, $V_h$ can be computed in polynomial time.

Under a fixed coloring $h$, we group colorful graphs $F$ in $\mathcal F_{K, b}$ according to their set $B(F)$ of branching vertices. Since $|B(F)|<b=O(1)$, there are at most $O(n^b) = \poly(n)$ such groups. Note that $V(F)\setminus B(F)$ decomposes into paths adjacent to the vertices of $B$ and cycles disjoint form $B$.
This allows us to evaluate the weighted sum in each group using dynamic programming in $\poly(n)e^{O(K)}$ time, see \Cref{subsec:color-dp}. In the regime $K=O(\log n)$ this gives the desired polynomial-time algorithm.

\paragraph{Organization of the paper.}
 In \Cref{sec:analysis}, we prove the theorem for pure SK model using the truncation and estimation bounds established in the subsequent sections. Specifically, \Cref{sec:truncation} bounds the error from restricting the graph expansion to even graphs with few edges and few branching vertices, while \Cref{sec:estimator} presents the deterministic color-coding algorithm. We analyze the estimator in \Cref{subsec:color-estimator}, derandomize the choice of colorings in \Cref{subsec:color-derandomization}, and show how to evaluate the estimator efficiently through weighted sums of paths and cycles in \Cref{subsec:color-dp}.

\Cref{app:mixed} extends the proof to mixed $p$-spin models. We derive the corresponding high-temperature expansion in \Cref{app:mixed-expansion}, establish the truncation bound in \Cref{app:mixed-truncation}, and extend the color-coding estimator and its dynamic programming implementation in \Cref{app:mixed-color-coding}. We then prove a high-probability lower bound on the normalized partition function in \Cref{app:mixed-lower-bound} and combine these ingredients to complete the proof in \Cref{app:mixed-completion}.

\section{Proof of the main theorem}\label{sec:analysis}
In this section, we specialize to the pure Sherrington--Kirkpatrick model. We prove the \Cref{thm:counting-intro} in this case  using the results established in subsequent sections. Recall from \Cref{sec:overview} that our main goal is to estimate the graph expansion 
$$
P(t) = \sum_{\substack{F\subseteq E \\ F\text{ even}}}t_F\,, \qquad \text{with }t_F = \prod_{\{i,j\}\in F} t_{ij}\,,
$$
where, for real inverse temperature $\beta$, we write $t_{ij} = t_{ij}(\beta) = \tanh(\beta g_{ij} / \sqrt{n})$. \Cref{sec:truncation} and \Cref{sec:estimator} construct $\widehat P$, which also takes the form of a multilinear polynomial in $t = (t_{ij})_{\{i,j\} \in E}$, $$\widehat P = \sum_{\substack{F\subseteq E\\F\text{ even}}} \widehat P_F t_F\,.$$  To reason about multilinear polynomials $f(t)=\sum_Ff_Ft_F$ such as $P$ and $\widehat P$,  we introduce a weighted norm below.
\begin{definition}\label{def:w_norm}
    For $w > 0$ and a multilinear polynomial $f(t)=\sum_Ff_Ft_F$, we define the weighted norm and inner product to be 
\begin{equation}\label{eq:coefficient-norm}
 \|f\|_w^2\coloneqq\sum_F|f_F|^2w^{|F|},\qquad
 \langle f,h\rangle_w\coloneqq\sum_F f_F\overline{h_F}w^{|F|}.
\end{equation}
\end{definition}

Throughout the proof, we fix constants $\beta_*$, $\lambda$ and $R_0$, such that
\begin{equation}\label{eq:w_def}
 \beta_*<R_0<\sqrt\lambda<1,\qquad \text{and define }w\coloneqq\lambda/n\,.
\end{equation}

Key observation is that for any $\beta$ such that $|\beta| < \beta_*$, the entries $t_{ij}(\beta)$ are  odd functions of the centered Gaussians $g_{ij}$, so for any $F \ne F'$, we have $\EE_g[t_F(\beta)t_{F'}(\beta)]=0$. Hence,
\begin{align*}
 \EE_g|f(t)|^2 &=\sum_{F,F'}f_F\overline{f_{F'}}
       \EE_g[t_F t_{F'}]
       =\sum_F|f_F|^2\prod_{\{i,j\}\in F}\EE_g\tanh^2(\beta g_{ij} / \sqrt{n})\\
 &\le\sum_F|f_F|^2(\beta^2/n)^{|F|} \leq \sum_F|f_F|^2 w^{|F|} = \|f\|_w^2\,,
\end{align*}
where the last inequality holds for any $w \geq \beta^2 / n$. In particular, we have that
$$
\EE_g |\widehat P(t(\beta)) - P(t(\beta))|^2 \leq \|\widehat P - P\|_w^2\,.
$$
In the remainder of the paper, we fix constants $\beta_*, R_0,$ and $\lambda$ such that $\beta_* < R_0 <\sqrt{\lambda} < 1$ and set 
\begin{equation*}
    w := \frac{\lambda}{n}.
\end{equation*}

Combined, \Cref{lem:sk-truncation} and \Cref{prop:sk-approximation} produce an estimate $\widehat P(t)$ satisfying, for every $L \geq \log n$, 
\begin{equation}\label{eq:comb_estimates}
    \|\widehat P - P\|_w^2 \leq e^{-L}\,.
\end{equation}
Moreover, $\widehat P$ can be computed in $\poly(n) e^{O(L)}$ arithmetic operations. The role of this section is to show that the approximation in \eqref{eq:comb_estimates} is sufficient for the multiplicative estimation of the partition function and yields the result in \Cref{thm:counting-intro}. 

Recall that the normalized partition function can be expressed as 
$$
Y_g(\beta) = \frac{\prod_{\{i,j\}\in E}\cosh(\beta g_{ij} / \sqrt{n})}{\EE_g Z_g(\beta)} \cdot P(t)\,,
$$
where the prefactors can be computed in polynomial time. Accordingly, for an estimate $\widehat P (t)$, our algorithm computes 
\begin{equation}\label{eq:y-hat-def}
\widehat Y_g(\beta) = \frac{\prod_{\{i,j\}\in E}\cosh(\beta g_{ij} / \sqrt{n})}{\EE_g Z_g(\beta)} \cdot \widehat P(t)\,.
\end{equation}

\begin{proposition}
\label{prop:coefficient-reduction}
Fix \(0<\beta_*<1\), and let \(w=\lambda/n\) be as above. There exist
constants $c > 0, C \geq 1$ and a sequence \(\zeta_n\to 0\), depending only on
\(\beta_*\), such that the following holds for all sufficiently large \(n\). Let \(0<\varepsilon,\delta<1/2\), and let
$
\widehat P(t)=\sum_{F\subseteq E}\widehat P_F t_F
$
be a deterministic multilinear polynomial whose coefficients are independent of $g$ and $\beta$, satisfying
\[
\|\widehat P-P\|_w^2
\le C^{-1} \varepsilon^2\delta^c.
\]
For each disorder realization $g$ and $|\beta|\leq \beta_*$, in additional $O(n^2)$ arithmetic operations, one can compute $\widehat Y_g(\beta)$ satisfying
\[
\mathbb P_g\left[
\sup_{|\beta|\le\beta_*}
\left|
\frac{\widehat Y_g(\beta)}{Y_g(\beta)}-1
\right|
\le \varepsilon
\right]
\ge 1-\zeta_n-\delta.
\]
\end{proposition}

Assuming this proposition, we first complete the proof of the main theorem for the SK model.

\begin{proof}[Proof of \Cref{thm:counting-intro} for Sherrington--Kirkpatrick Model]
    There exists a constant $\kappa> 0$, such that for any $\log n \leq L \leq \kappa n$, \Cref{lem:sk-truncation} and \Cref{prop:sk-approximation}, applied with an appropriate choice of $L' = O(L)$, can be combined via a triangle inequality to produce an estimate $\widehat P$ satisfying 
    $$
    \| \widehat P - P\|_w^2 \leq e^{-L}\,.
    $$
    Moreover, $\widehat P$ can be computed in $\poly(n)e^{O(L)}$ arithmetic operations. Let $c, C > 0$ be the constants in \Cref{prop:coefficient-reduction}, and choose 
    $$
    L = \lceil \max \{\log n, \log C + 2 \log(1/\varepsilon) + c \log (1/\delta) \} \rceil\,.
    $$
    Then $e^{-L} \leq C^{-1} \varepsilon^2 \delta^c$, so \Cref{prop:coefficient-reduction} implies 
    \[
    \mathbb P_g\left[
    \sup_{|\beta|\le\beta_*}
    \left|
    \frac{\widehat Y_g(\beta)}{Y_g(\beta)}-1
    \right|
    \le \varepsilon
    \right]
    \ge 1-\zeta_n-\delta.
    \]
    Moreover, $\poly(n) e^{O(L)} + O(n^2) = \poly(n, \varepsilon^{-1}, \delta^{-1})$, which is the running time claimed in \Cref{thm:counting-intro}.
\end{proof}

\medskip
The proof of \Cref{prop:coefficient-reduction} in \Cref{subsec:weighted_norm_to_mult_guarantee} has two parts. The first (\Cref{lem:entire-lift}) converts coefficient error intro a uniform additive error: 
$$
\EE_g \sup_{|\beta| \leq \beta_*} \left|\widehat Y_g(\beta) - Y_g(\beta)\right|^2 \leq C \|\widehat P - P\|^2_w\,.
$$
The second (\Cref{lem:lower-bound}, proved in \Cref{subsec:unif_lb_normalized_partition}) gives a uniform high probability lower bound on the modulus of the normalized partition function $Y_g$. This follows from the zero-freeness theorem of \cite{BHL25}, a second-moment upper bound in \Cref{fact:overlap-moment}, and Harnack's inequality. Markov's inequality then yields \Cref{prop:coefficient-reduction}.

\subsection{Controlling the approximation error}\label{subsec:weighted_norm_to_mult_guarantee}

For deterministic multilinear $f(t)=\sum_Fa_Ft_F$, define
\begin{equation}\label{eq:tr-lift}
 \mathcal L_gf(z)\coloneqq e^{-(n-1)z^2/4}
 \sum_Fa_F\prod_{\{i,j\}\in F}\sinh(zg_{ij}/\sqrt{n})
                  \prod_{\{i,j\}\in E\setminus F}\cosh(zg_{ij}/\sqrt{n}).
\end{equation}

\begin{lemma}\label{lem:entire-lift}
$\mathcal L_gf(z)$ is entire and $\sup_{|z|\le R_0}\EE_g|\mathcal L_gf(z)|^2
 \le e^{R_0^4/4}\|f\|_w^2$.

\end{lemma}
\begin{proof}
Entirety follows from the finite sum of products of entire functions.

For a standard Gaussian $g_{ij}$, writing $z=s+it$ gives
\[
 \EE|\cosh(zg_{ij}/\sqrt{n})|^2=e^{(s^2-t^2)/n}\cosh(|z|^2/n),\qquad
 \EE|\sinh(zg_{ij}/\sqrt{n})|^2=e^{(s^2-t^2)/n}\sinh(|z|^2/n).
\]
Therefore, 
\begin{equation}\label{eq:tr-lift-norm}
 \EE_g|\mathcal L_gf(z)|^2
 =\cosh(|z|^2/n)^{n(n-1)/2}
       \sum_F|a_F|^2\tanh(|z|^2/n)^{|F|}.
\end{equation}

Moreover, for $|z|\le R_0$, the inequalities $\tanh x\le x$ and
$\log\cosh x\le x^2/2$ for $x\ge0$ give
\[
 \tanh(|z|^2/n)\le R_0^2/n\le w,\qquad
 \cosh(|z|^2/n)^{n(n-1)/2}
 \le\exp\left(\frac{R_0^4(n-1)}{4n}\right)
 \le e^{R_0^4/4}\,,
\]
where recall that we chose the constant $R_0$ to satisfy $\beta_* < R_0 < \sqrt{\lambda}$.
Plugging these bounds into \eqref{eq:tr-lift-norm} proves the lemma.
\end{proof}

The graph expansion proves $\mathcal L_gP=Y_g$. Denote
$f\coloneqq\widehat P-P$ and note that $\mathcal L_gf=\mathcal L_g\widehat P-Y_g$. 
\begin{lemma}\label{lem:cauchy-sup}
    There exists a constant $C_{\mathrm{err}}>0$, such that
    $$
    \EE_g\sup_{|z|\le\beta_*}|\mathcal L_g\widehat P(z)-Y_g(z)|^2 \le C_{\mathrm{err}}\|\widehat P-P\|_w^2\,.
    $$
\end{lemma}
\begin{proof} We use the following standard consequence of Cauchy's integral formula. If $F$ is holomorphic on a neighborhood
of the closed disk $\{|z|\le T\}$ and $0 < R < T$, 
\begin{equation}\label{eq:an-cauchy-sup1}
 \sup_{|z|\le R}|F(z)|^2
 \le (1-R/T)^{-2}
 \int_0^{2\pi}|F(Te^{i\theta})|^2\,\frac{d\theta}{2\pi}.
\end{equation}
    Apply \eqref{eq:an-cauchy-sup1} to the entire function $\mathcal L_gf(z)$ with
$R=\beta_*$ and $T=R_0$, and take expectations. The bound
in \Cref{lem:entire-lift}, with constant $C_0$, then gives
\begin{equation}\label{eq:coefficient-uniform}
 \begin{aligned}
 \EE_g\sup_{|z|\le\beta_*}|\mathcal L_g\widehat P(z)-Y_g(z)|^2=\EE_g\sup_{|z|\le\beta_*}|\mathcal L_gf(z)|^2
 &\quad\le(1-\beta_*/R_0)^{-2}
   \int_0^{2\pi}\EE_g|\mathcal L_gf(R_0e^{i\theta})|^2
       \frac{d\theta}{2\pi}\\
 &\quad\le C_0(1-\beta_*/R_0)^{-2}\|f\|_w^2
 \le C_{\mathrm{err}}\|\widehat P-P\|_w^2.
 \end{aligned}
\end{equation}
\end{proof}

We use a lower bound on the normalized partition function $Y_g(z)$, uniform over the complex disk $|z|\le\beta_*$. This lemma is proved in the subsequent \Cref{subsec:unif_lb_normalized_partition}.
\begin{lemma}\label{lem:lower-bound}
Fix a radius $R_1$ with $\beta_*<R_1<1$. There are constants $c,C>0$ such that, for every $0<\alpha<1/2$,
\[
 \PP_g\!\left[\inf_{|z|\le\beta_*}|Y_g(z)|\ge C^{-1}\alpha^c\right]
 \ge1-\zeta_n-\alpha\,,
\]
where $\zeta_n\to0$ depends only on $R_1$ and $n$.
\end{lemma}

Fix $0<\varepsilon,\delta<1/2$ and put $\alpha\coloneqq\delta/2$,
so that each of the two estimates below may fail with probability
at most $\alpha$.
Let $C_{\mathrm{err}}\ge1$ denote the constant in \Cref{lem:cauchy-sup} and let $c, C>0$ be constants from \Cref{lem:lower-bound}.
By Markov's inequality and \Cref{lem:cauchy-sup},

 \begin{equation}\label{eq:markov_final_mult}
 \PP_g\!\left[
   \sup_{|z|\le\beta_*}|\mathcal L_g\widehat P(z)-Y_g(z)|
      >\varepsilon C^{-1}\alpha^c\right]\le\frac{\EE_g\sup_{|z|\le\beta_*}
          |\mathcal L_g\widehat P(z)-Y_g(z)|^2}
          {\varepsilon^2C^{-2}\alpha^{2c}}\le\frac{C_{\mathrm{err}}\|\widehat P-P\|_w^2}{\varepsilon^2C^{-2}\alpha^{2c}}\leq\alpha\,,
 \end{equation}
where the last inequality comes from choosing the constants in \Cref{prop:coefficient-reduction} so that $\|\widehat P-P\|_w^2 \leq C^{-2} C_{\mathrm{err}}^{-1} \varepsilon^2 \alpha^{1+2c}$.

On the event that both \Cref{lem:lower-bound} and the bound in \eqref{eq:markov_final_mult} hold, which happens with probability at least
$1-\zeta_n-2\alpha=1-\zeta_n-\delta$,
\[
 \sup_{|\beta|\le\beta_*}
 \left|\frac{\mathcal L_g\widehat P(\beta)}{Y_g(\beta)}-1\right|
 \le\frac{\varepsilon C^{-1}\alpha^c}{C^{-1}\alpha^c}=\varepsilon.
\]
For real $\beta$, $\mathcal L_g\widehat P(\beta) = \widehat Y_g(\beta)$, as defined in \eqref{eq:y-hat-def}, which concludes the proof.

\subsection{Lower bound for the normalized partition function}\label{subsec:unif_lb_normalized_partition}

\begin{mylem}{\ref{lem:lower-bound}}
Fix a radius $R_1$ with $\beta_*<R_1<1$. There are constants $c,C>0$ such that, for every $0<\alpha<1/2$,
\[
 \PP_g\!\left[\inf_{|z|\le\beta_*}|Y_g(z)|\ge C^{-1}\alpha^c\right]
 \ge1-\zeta_n-\alpha\,,
\]
where $\zeta_n\to0$ depends only on $R_1$ and $n$.
\end{mylem}

The proof of \Cref{lem:lower-bound} at the end of this section combines a second-moment bound, Cauchy's integral formula, and Markov's
inequality into a uniform upper bound on $|Y_g|$ in the larger
disk $|z|\le R_1$ with high probability (\Cref{lem:upper-bound}).
Harnack's inequality then combines this upper bound with
zero-freeness (\Cref{fact:zero-free}) and the normalization $Y_g(0)=1$ to yield the
desired lower bound on $|z|\le\beta_*$.

\medskip

The zero-free theorem of Bencs et al.~\cite{BHL25} gives the following proposition; we verify the change of the normalization below.
\begin{proposition}[\cite{BHL25}, Theorem 1.2]\label{fact:zero-free}
The function $Y_g$ has no zeros in $|z|<R_1$ with probability at
least $1-\zeta_n$, where $\zeta_n\to0$ depends only on $R_1$ and $n$.
\end{proposition}
Indeed, their ordered-tuple SK Hamiltonian can be written as
\begin{equation}\label{eq:bhl-normalization}
 \widetilde H_G(\sigma)
 =\frac1{\sqrt{2n}}\sum_{i,j=1}^nG_{ij}\sigma_i\sigma_j
 =H_g(\sigma)+X,\qquad
 g_{ij}\coloneqq\frac{G_{ij}+G_{ji}}{\sqrt2},\quad
 X\coloneqq\frac1{\sqrt{2n}}\sum_iG_{ii}.
\end{equation}
Here the $G_{ij}$ are independent standard Gaussians, so the
$g_{ij}$, $i<j$, are independent standard Gaussians and
$X\sim N(0,1/2)$ is independent of them. In particular,
\begin{equation}\label{eq:bhl-partition-normalization}
 \widetilde Z_G(z)=e^{zX}Z_g(z),\qquad
 \widetilde Y_G(z)\coloneqq e^{-nz^2/4}\widetilde Z_G(z)
 =e^{zX-z^2/4}Y_g(z).
\end{equation}
These multipliers never vanish, so the zero-free theorem of \cite{BHL25} transfers to $Y_g$ without loss. 

\medskip

\begin{lemma}\label{fact:overlap-moment}
For every $0\le T<1$ and $|z|\le T$,
\begin{equation}\label{eq:an-second-moment}
 \EE_g|Y_g(z)|^2
 \le\frac{e^{-|z|^2/2}}{\sqrt{1-|z|^2}}
 \le\frac1{\sqrt{1-T^2}},
\end{equation}

\end{lemma}
\begin{proof}
Let $\sigma,\tau$ be independent uniform spin configurations and
let $\rho_i=\sigma_i\tau_i$, so that $\rho_i$ are independent uniform
signs. Gaussian integration, with the self terms canceled by the
normalization, gives
\[
\begin{aligned}
 \EE_g|Y_g(z)|^2
 =\EE_\rho\exp\left(
       \frac{|z|^2}{n}\sum_{i<j}\rho_i\rho_j\right)=e^{-|z|^2/2}\EE_\rho\exp\left(
       \frac{|z|^2}{2n}\Bigl(\sum_i\rho_i\Bigr)^2\right),
\end{aligned}
\]
where we used
$2\sum_{i<j}\rho_i\rho_j=(\sum_i\rho_i)^2-n$.

For $|z|^2<1$ and an independent $h\sim N(0,1)$,
Gaussian integration and $\cosh x\le e^{x^2/2}$ give
\[
 \EE_\rho\exp\left(\frac {|z|^2}{2n}\Bigl(\sum_i\rho_i\Bigr)^2\right)
 =\EE_h\cosh\left(|z| h / \sqrt{n} \right)^n
 \le\EE_he^{|z|^2 h^2/2}=\frac{1}{\sqrt{1-|z|^2}}.
\]
This proves the bound for every $n$.
\end{proof}

The second-moment bound in \Cref{fact:overlap-moment} and Cauchy's integral formula give the
following uniform upper bound with high probability.

\begin{lemma}
\label{lem:upper-bound}
Fix $0<R_1<1$. There is a constant $C_U\ge1$, depending only
on $R_1$, such that, for every $0<\alpha<1/2$,
\[
 \PP_g\!\left[
   \sup_{|z|\le R_1}|Y_g(z)|\le\sqrt{C_U/\alpha}
 \right]\ge1-\alpha.
\]
\end{lemma}

\begin{proof}
First, let $0\le R<T$ and let $F$ be holomorphic on a neighborhood
of the closed disk $\{|z|\le T\}$. For $|z|\le R$,
Cauchy's integral formula gives
\[
 F(z)=\int_0^{2\pi}F(Te^{i\theta})
       \frac{Te^{i\theta}}{Te^{i\theta}-z}\,
       \frac{d\theta}{2\pi}.
\]
By Cauchy--Schwarz and $|Te^{i\theta}-z|\ge T-R$,
\begin{equation}\label{eq:an-cauchy-sup}
 \sup_{|z|\le R}|F(z)|^2
 \le (1-R/T)^{-2}
 \int_0^{2\pi}|F(Te^{i\theta})|^2\,\frac{d\theta}{2\pi}.
\end{equation}

Now set $T=(1+R_1)/2$, so that $R_1<T<1$.
Since $Y_g$ is entire, we may apply
\eqref{eq:an-cauchy-sup} with $F=Y_g$ and $R=R_1$.
Taking expectations, interchanging expectation and integration
by Tonelli's theorem, and applying
\Cref{fact:overlap-moment}, we obtain
\[
 \EE_g\sup_{|z|\le R_1}|Y_g(z)|^2
 \le (1-R_1/T)^{-2}
 \int_0^{2\pi}\EE_g|Y_g(Te^{i\theta})|^2\,
       \frac{d\theta}{2\pi}
 \le \frac{(1-R_1/T)^{-2}}{\sqrt{1-T^2}}
 \coloneqq C_U.
\]
This constant satisfies $C_U\ge1$ and depends only on $R_1$.
Setting $M\coloneqq\sqrt{C_U/\alpha}$, Markov's inequality gives
\[
 \PP_g\!\left[\sup_{|z|\le R_1}|Y_g(z)|>M\right]
 \le \frac{\EE_g\sup_{|z|\le R_1}|Y_g(z)|^2}{M^2}
 \le \alpha,
\]
which proves the claim.
\end{proof}

\begin{proof}[Proof of Lemma~\ref{lem:lower-bound}]
Let $C_U$ be the constant from Lemma~\ref{lem:upper-bound}
and set $M\coloneqq\sqrt{C_U/\alpha}$.
By that lemma, \Cref{fact:zero-free}, and the union
bound, with probability at least $1-\zeta_n-\alpha$,
the function $Y_g$ has no zeros in $|z|<R_1$ and satisfies
$\sup_{|z|\le R_1}|Y_g(z)|\le M$.

On this event, the function
\[
 u(z)\coloneqq\log M-\log|Y_g(z)|
\]
is nonnegative and harmonic in $|z|<R_1$.
Since $Y_g(0)=1$, we have $u(0)=\log M$.
Harnack's inequality therefore gives
\[
 u(z)\le\frac{R_1+\beta_*}{R_1-\beta_*}\log M
 \qquad (|z|\le\beta_*).
\]
Rearranging and exponentiating, we obtain
\[
 \inf_{|z|\le\beta_*}|Y_g(z)|
 \ge M^{-2\beta_*/(R_1-\beta_*)}
 = C_U^{-c}\alpha^c,
 \qquad c\coloneqq\frac{\beta_*}{R_1-\beta_*}>0.
\]
Taking $C=C_U^c$ proves the lemma.
\end{proof}

\section{Truncation bounds}\label{sec:truncation}

As outlined in \Cref{sec:overview}, we approximate $P$ by truncating the expansion to $P_{K,b}$. This section upper bounds the resulting squared error in the weighted coefficient norm\footnote{Defined in \Cref{sec:analysis}, \eqref{eq:coefficient-norm}} (\Cref{lem:sk-truncation}). Concretely, for a graph $F \subseteq E$ denote the set of supported vertices $V(F)\coloneqq\{i:\deg_F(i)>0\}$ and the set of \emph{branching vertices}, i.e. those of degree at least $4$, and their number as
\begin{align*}
    B(F)\coloneqq\{i:\deg_F(i)\ge4\},\qquad s(F)\coloneqq|B(F)|\,.
\end{align*}
We truncate the expansion $P(t) = \sum_{\substack{F\subseteq E\\F\text{ even}}}t_F$ by only keeping the terms with at most $K$ edges and fewer than $b$ branching vertices. Define the retained family and its polynomial by
\[
 \mathcal F_{K,b}\coloneqq\{F\subseteq E:F\text{ even},\ |F|\le K,\ s(F)<b\},
 \qquad P_{K,b}(t)\coloneqq\sum_{F\in\mathcal F_{K,b}}t_F.
\]
Recall that the edge weight $w$ satisfies $w=\lambda/n$ for a fixed $0<\lambda<1$ (\eqref{eq:w_def}).
\begin{lemma}\label{lem:sk-truncation}
There exists $\kappa>0$ such that

for $n\ge2$ and $\log n\le L\le\kappa n$, there are integer cutoffs $K \geq 0, b \geq 1$ with $K=O(L)$ and $b = O \left(\frac L {\log(n/L)}\right)$, satisfying
\begin{equation}\label{eq:graph-cutoffs}
 \|P-P_{K,b}\|_w^2\le e^{-L}.
\end{equation}
All constants depend only on $\lambda$.
\end{lemma}
\begin{proof}
We expand
    \begin{equation}\label{eq:truncation_ub_two_terms}
    \|P-P_{K,b}\|_w^2 = \sum_{\substack{F \subseteq E \text{ even}\\ F \not\in \mathcal F_{K,b}}} w^{|F|} \le\sum_{\substack{F\text{ even}\\|F|>K}}w^{|F|}
    +\sum_{\substack{F\text{ even}\\s(F)\ge b}}w^{|F|}
    \end{equation}
    and bound the two terms separately. First, by an argument analogous to that in the proof of \Cref{fact:overlap-moment}, we upper bound the total weighted norm $\|P\|_w^2$. Let $\rho_1,\dots,\rho_n$ be independent uniform signs:
    \begin{equation}\label{eq:graph-mass}
        \|P\|_w^2 = \EE_\rho \prod_{i < j} (1 + w \rho_i \rho_j) \leq \EE_{\rho} \exp\left( w \sum_{i < j} \rho_i \rho_j\right) = e^{-\lambda/2} \EE_{\rho} \exp\left( \frac{\lambda}{2n} \left(\sum_{i} \rho_i\right)^2\right) \leq \frac{1}{\sqrt{1-\lambda}}\,.
    \end{equation}
    Choose $\theta>1$ with $\theta\lambda<1$ and apply the above bound with weight $\theta w$ to obtain
\begin{equation}\label{eq:graph-edge-tail}
 \sum_{\substack{F\text{ even}\\|F|>K}}w^{|F|}
 \le\theta^{-K}(1-\theta\lambda)^{-1/2}
 = Ce^{-cK} \le e^{-L}/2\,,
\end{equation}
where $c$ and $C$ are constants depending only on $\lambda$ and the last inequality comes from choosing a sufficiently large $K=O(L)$ for a given $L$.

To bound the second term, fix a particular set $S \subseteq [n]$ of $b$ vertices and let 
$$
M_S \coloneq \sum_{\substack{ F \subseteq E \text{ even}\\ \deg_F(i)\geq 4 \forall i \in S }} w^{|F|}\,,
$$
so that $\sum_{\substack{F\text{ even}\\s(F)\ge b}}w^{|F|} \leq \sum_{S \in \binom{[n]}{b}} M_S$. Fix some $\theta > 1$ satisfying $w\theta^2 < 1$ and let $q_{ij} = w \theta^{\ind\{i\in S\} + \ind\{j\in S\}}$. Then, again letting $\rho_1,\dots,\rho_n$ be independent uniform signs, we can bound
\begin{align*}
    M_S \theta^{4b} \leq \sum_{F\subseteq E \text{ even}} \prod_{\{i,j\} \in F} q_{ij} \leq \EE_\rho \exp \left( \sum_{i<j} q_{ij} \rho_i \rho_j \right) \leq \EE_\rho \exp \left( \frac w 2 \left( \sum_i \rho_i \theta^{\ind\{i\in S\}} \right)^2 \right)\,.
\end{align*}
For an independent $h \sim \mathcal N(0,1)$, we can further bound the right hand side as 
\begin{align*}
    \EE_\rho \exp \left( \frac w 2 \left( \sum_i \rho_i \theta^{\ind\{i\in S\}} \right)^2 \right) &= \EE_h \EE_\rho \exp \left( \sqrt{w} h \left( \sum_i \rho_i \theta^{\ind\{i\in S\}} \right) \right)\\
    &= \EE_h \cosh (\sqrt{w} \theta h)^b \cosh(\sqrt{w} h)^{n-b}\\
    &\leq \EE_h \exp \left( \frac w 2 (b \theta^2 + n-b) h^2 \right) = (1 - w (b \theta^2 + n-b))^{-1/2}\,.
\end{align*}
The last step is valid for the choice of $\theta = \sqrt{1 + \frac {c'n}{b}}$ for a constant $c'>0$, such that $w (b \theta^2 + n-b) = \lambda (c'+1) < 1$. Note that this choice of $\theta$ also satisfies $\theta > 1$ and $w \theta^2 < 1$ for sufficiently small $c'$. We conclude by bounding 
$$\sum_{\substack{F\text{ even}\\s(F)\ge b}}w^{|F|} \leq \sum_{S \in \binom{[n]}{b}} M_S \leq \binom n b (1 - \lambda (c'+1))^{-1/2} \theta^{-4b} \leq (1 - \lambda (c'+1))^{-1/2} \left( \frac{en}{b \theta^4} \right)^b \le C(Cb/n)^b$$
for a choice of constant $C$ only dependent on $\lambda$. Finally, set $b \coloneq \lceil \frac {L + \log(2C)}{\log(n/ (CL + C\log(2C)))}\rceil$, which is $\leq L + \log (2C)$ for sufficiently small $\kappa$. Then, 
$
b \log \frac{n}{C b} \geq  L + \log(2C)
$, so 
\begin{equation}\label{eq:graph-branching_tail}
\sum_{\substack{F\text{ even}\\s(F)\ge b}}w^{|F|} \leq C(Cb/n)^b \leq C e^{-L - \log (2C)} = e^{-L}/2\,.
\end{equation}
Plugging the bounds in \eqref{eq:graph-edge-tail} and \eqref{eq:graph-branching_tail} into \eqref{eq:truncation_ub_two_terms}, we obtain 
\[
 \|P-P_{K,b}\|_w^2
 \le\sum_{\substack{F\text{ even}\\|F|>K}}w^{|F|}
    +\sum_{\substack{F\text{ even}\\s(F)\ge b}}w^{|F|}
 \le e^{-L}.
\]
\end{proof}

\section{Deterministic estimator via color coding}
\label{sec:estimator}
In this section, we construct and analyze a deterministic estimator for the truncated graph expansion from \Cref{sec:truncation}. Our goal is to find an efficiently computable estimator $\widehat P$
such that the squared weighted coefficient error
$\|\widehat P-P_{K,b}\|_w^2$ is small. Write
$\mathcal F\coloneqq\mathcal F_{K,b}$ and
$P_{\mathcal F}\coloneqq P_{K,b}$.

\begin{proposition}
\label{prop:sk-approximation}
For every integer $L\ge\log n$, one can construct an estimator $\widehat P_{\mathcal F}$ satisfying
\[
 \|\widehat P_{\mathcal F}-P_{\mathcal F}\|_w^2\le e^{-L}.
\]
Moreover, $\widehat P_{\mathcal F}$ can be computed in $\poly(n)e^{O(L)}$ arithmetic operations. 
\end{proposition}

We begin with a sketch of the argument.  Naively evaluating $P_{\mathcal F}$ by enumerating its terms takes
$n^{O(K)}$ operations, giving only a quasipolynomial algorithm when
$K=\Theta(\log n)$. To overcome this issue, we use color coding method,
introduced by Alon, Yuster, and Zwick~\cite{AYZ95}.

To sketch the method, consider a uniformly random balanced coloring $h:[n]\to[K]$ (balanced means that the color classes have sizes differing by at most one). We call a graph $F$ \emph{colorful} under $h$ if its vertices $V(F)$ have different colors. We then consider the sum
\begin{equation}\label{eq:color-estimator-sketch}
 V_h\coloneqq\sum_{F\in\mathcal F}
 \frac{\ind\{F\text{ colorful under }h\}}
      {p_{|V(F)|}}t_F.
\end{equation}
Here $p_v$ is the probability that a fixed vertex set of size $v$
is colorful. This makes sure that
\[
 \EE_h V_h=P_{\mathcal F}.
\]

Although $V_h$ may still contain $n^{O(K)}$ terms, colorfulness
allows us to evaluate it without enumerating those terms.
In \Cref{subsec:color-dp}, we construct a dynamic program that
computes $V_h$ exactly using
\[
 \poly(n)e^{O(K)}\sum_{j=0}^{b-1}\binom nj
\]
operations. This is polynomial in $n$ when
$K=O(\log n)$ and $b\log(e n/b)=O(\log n)$, which is the regime we are working in.

Once we can evaluate the sum for one coloring, we can repeat
the procedure for $T$ independent balanced colorings
$h_1,\ldots,h_T$ and set
\[
 \widehat P_{\mathcal F}\coloneqq\frac1T\sum_{t=1}^T V_{h_t}.
\]
Unbiasedness and independence imply
\begin{equation}\label{eq:color-average-error-sketch}
 \EE_{h_1,\ldots,h_T}
 \|\widehat P_{\mathcal F}-P_{\mathcal F}\|_w^2 =
 \frac1T\EE_h\|V_h-P_{\mathcal F}\|_w^2\le
 \frac{O(8^K)}{T}
\end{equation}
where the inequality follows from a simple computation. Thus, for any $\eta>0$, choosing $T \gtrsim 8^K/\eta$ makes the expected
squared coefficient error at most $\eta$.

This sketches a randomized algorithm for approximating
$P_{\mathcal F}$. In our real algorithm, however, the colorings are
chosen deterministically while preserving the same error bound.
We use the method of conditional expectations
(see, for example,~\cite{ES73, Rag88}).

The section is organized as follows.
In \Cref{subsec:color-estimator}, we reduce the proof of
\Cref{prop:sk-approximation} to constructing suitable colorings
deterministically and efficiently evaluating their associated polynomials $V_h$. In \Cref{subsec:color-derandomization}, we construct these colorings using the method of conditional expectations.
Finally, in \Cref{subsec:color-dp}, we give a dynamic program for
evaluating $V_h$.

\subsection{The reduction}
\label{subsec:color-estimator}

In this section, we will prove \Cref{prop:sk-approximation} given \Cref{lem:color-overlap} and \Cref{lem:colorful-evaluator} below.

\begin{lemma}
\label{lem:color-overlap}
For every integer $T\ge1$, we can deterministically construct
balanced colorings $h_1,\ldots,h_T$ such that, for each
$t=1,\ldots,T$,
\begin{equation}\label{eq:det-greedy}
 \sum_{i<t}\langle V_{h_i},V_{h_t}\rangle_w
 \le (t-1)\|P_{\mathcal F}\|_w^2..
\end{equation}
The construction uses $\poly(n,T)e^{O(K)}
 \sum_{j=0}^{b-1}\binom nj$ operations.
\end{lemma}

We prove \Cref{lem:color-overlap} in
\Cref{subsec:color-derandomization} with an efficient evaluator from \Cref{subsec:color-dp} as an input.
\begin{lemma}
\label{lem:colorful-evaluator}
Given any coloring $h:[n]\to[K]$, the polynomial $V_h$ can
be evaluated exactly using
\[
 \poly(n)e^{O(K)}
 \sum_{j=0}^{b-1}\binom nj
\]
arithmetic operations.
\end{lemma}

We prove \Cref{lem:colorful-evaluator} in
\Cref{subsec:color-dp}. 

\begin{proof}[Proof of \Cref{prop:sk-approximation}]
We first lower bound on $p_v$. Let $n_1,\ldots,n_K$ be the color-class sizes. Since they are balanced, we have $n_c\ge n/(2K)$, so
\[
 p_v
 =   {\binom nv}^{-1}\sum_{\substack{S\subseteq[K] \\ |S|=v}}
   \prod_{c\in S}n_c
 \ge
 {\binom nv}^{-1}{\binom Kv \left(\frac{n}{2K}\right)^{v}}
 \ge \frac{(K)_v}{2^vK^v}
 \ge \frac{K!}{2^KK^K}
 \ge (2e)^{-K}
 >8^{-K},
\]
Next, note that every balanced coloring retains a fraction $p_v$ of the
supports of size $v$, so the normalization in
\eqref{eq:color-estimator-sketch} gives
\[
 \langle V_h,P_{\mathcal F}\rangle_w
 =\|P_{\mathcal F}\|_w^2.
\]
Since the coefficients of $V_h$ lie in $[0,8^K]$ and its
monomials belong to $\mathcal F$, we also have
\[
 \|V_h\|_w^2
 \le 8^K\langle V_h,P_{\mathcal F}\rangle_w
 =8^K\|P_{\mathcal F}\|_w^2.
\]

Take the colorings supplied by \Cref{lem:color-overlap}
and consider $\widehat P_{\mathcal F}$ as above. 
First observe that by the identity above, $\langle\widehat P_{\mathcal F},P_{\mathcal F}\rangle_w
 =\frac1T\sum_{t=1}^T
   \langle V_{h_t},P_{\mathcal F}\rangle_w
 =\|P_{\mathcal F}\|_w^2.$ Hence, expanding the squared error
\[
\begin{aligned}
 \|\widehat P_{\mathcal F}-P_{\mathcal F}\|_w^2
 = \|\widehat P_{\mathcal F}\|_w^2 - \|P_{\mathcal F}\|_w^2 
 &=
 \frac1{T^2}\sum_{t=1}^T\|V_{h_t}\|_w^2
 +\frac2{T^2}\sum_{t=1}^T\sum_{i<t}
   \langle V_{h_i},V_{h_t}\rangle_w
 -\|P_{\mathcal F}\|_w^2\\
 &\le
 \left(
 \frac{8^K}{T}
 +\frac2{T^2}\sum_{t=1}^T(t-1)-1
 \right)\|P_{\mathcal F}\|_w^2=
 \frac{8^K-1}{T}\|P_{\mathcal F}\|_w^2
 \le \frac{C8^K}{T}.
\end{aligned}
\]
where the last inequality uses
$\|P_{\mathcal F}\|_w^2\le\|P\|_w^2\le C$
from \eqref{eq:graph-mass}. Hence, choosing $T=C8^K e^L$  gives the required error bound.

Finally, \Cref{lem:color-overlap} constructs the colorings
deterministically, and \Cref{lem:colorful-evaluator} evaluates
each associated polynomial. The total construction and
evaluation cost is
\[
 \poly(n,T)e^{O(K)}
 \sum_{j=0}^{b-1}\binom nj.
\]
Our cutoff bounds $K=O(L)$, $\log T=O(L)$, and $\sum_{j=0}^{b-1}\binom nj
 \le \exp \{O(b\log(e n/b))\}
 =e^{O(L)},$ imply that the total number of operations is at most $\poly(n)e^{O(L)}$.
\end{proof}

\subsection{Deterministic selection of colorings}
\label{subsec:color-derandomization}

In this section, we show how to control the  overlap of each new coloring with the colorings already chosen.
A uniformly random balanced coloring typically has small overlap with previous colors. We turn this observation into a deterministic construction by assigning colors vertex by vertex, always preserving a bound on the conditional expectation. The main idea is that these conditional expectations can be computed efficiently by dynamic programming from \Cref{subsec:color-dp}.

\begin{proof}[Proof of \Cref{lem:color-overlap}]
Fix color-class sizes $n_1,\ldots,n_K$ differing by at most one
and summing to $n$. All random colorings below are uniform among
colorings with these prescribed sizes.

Suppose $h_1,\ldots,h_{t-1}$ have already been chosen. Define
\[
 S(h)\coloneqq\sum_{i<t}\langle V_{h_i},V_h\rangle_w,
\]
and let $H$ be a fresh random coloring. Unbiasedness and the
identity $\langle V_h,P_{\mathcal F}\rangle_w
=\|P_{\mathcal F}\|_w^2$ from the preceding subsection give
\begin{equation}\label{eq:color-initial-overlap}
 \EE_H S(H)
 =\sum_{i<t}\langle V_{h_i},P_{\mathcal F}\rangle_w
 =(t-1)\|P_{\mathcal F}\|_w^2.
\end{equation}
Thus it suffices to choose $h$ whose overlap is at most this
average. 

We now assign colors to the vertices one by one in the order $1,\ldots,n$, while maintaining the bound on the conditional expectation. Initially $\chi_0 = \varnothing$ is an empty coloring. For a partial assignment $\chi_k:[k]\to[K]$ that has a balanced completion, write
\[
 \Phi(\chi_k)\coloneqq
 \EE_H\!\left[S(H)\mid H|_{[k]}=\chi_k\right].
\]
Note that using the law of total expectation, we have
\begin{align*}
    \Phi_t(\chi_k) &= \sum_{c \in [K]} \mathbb{P}_H\!\left[H(k+1)=c\mid H|_{[k]}=\chi_k\right]\times \EE_H\!\left[S(H)\mid H|_{[k]}=\chi_k \text{ and } H(k+1) = c \right] \\
    & \geq \min_{c \in [K]}\EE_H\!\left[S(H)\mid H|_{[k]}=\chi_k \text{ and } H(k+1) = c \right] 
\end{align*}
Therefore, we extend $\chi_k$ from $[k]$ to $[k+1]$ by choosing $\chi_{k+1}(k+1)=c$, for the color $c$ with the smallest conditional expectation, so we guarantee that 
\begin{equation*}
    \Phi(\chi_{k+1}) \leq \Phi(\chi_k).
\end{equation*}
Note that if at every step, we can compute the conditional expectations exactly and efficiently, we would guarantee that at the end of this process, we have a deterministic coloring $\chi_n$, that satisfies $S(\chi_n) \leq (t-1)\|P_{\mathcal F}\|_w^2$.

To this end, it suffices that given a coloring $h:[n]\to[K]$ and a partial coloring
$\chi:[k]\to[K]$ admitting a balanced completion, we can compute $\EE_H[\langle V_h,V_H\rangle_w
              \mid H|_{[k]}=\chi]$ efficiently. Note that 
\begin{equation*}
  \EE_H[\langle V_h,V_H\rangle_w
              \mid H|_{[k]}=\chi] =   \sum_{\substack{F\in\mathcal F\\F\text{ colorful under }h}}
   \frac{w^{|F|}}{p_{|V(F)|}^2}
   \PP_H\!\left[H\text{ injective on }V(F)
                  \mid H|_{[k]}=\chi\right].
\end{equation*}
The probabilities $\PP_H[H\text{ injective on }V(F)\mid H|_{[k]}=\chi]$ only depend on the colors $C = \chi( V(F) \cap [k] )$ and the size $\ell =|V(F) \setminus k|$, and they are explicitly computable. This means that the conditional expectation, can be expressed as
\begin{equation}\label{eq:conditional-overlap}
    \EE_H[\langle V_h,V_H\rangle_w
              \mid H|_{[k]}=\chi] =   \sum_{\substack{F\in\mathcal F\\F\text{ colorful under }h}} \alpha (|V(F)|, |V(F) \cap [k]|, \chi(V(F) \cap [k]))\, w ^{|F|},
\end{equation}
where $\alpha : [n] \times [n] \times 2^{[K]} \to \mathbb{R}$ is an explicit coefficients that can be computed in $\poly(n)e^{O(K)}$ time. Therefore, \Cref{prop:colorful-coefficients}, applied with
$U=[k]$, coefficient table $\alpha$, and edge inputs $t_{ij}=w$,
computes each conditional overlap exactly using
\[
 \poly(n)e^{O(K)}\sum_{j=0}^{b-1}\binom nj
\]
arithmetic operations. This concludes the proof.
\end{proof}

\subsection{Dynamic programming for a fixed coloring}
\label{subsec:color-dp}

We have now established the correctness of the algorithm. It only remains to prove that we can efficiently evaluate the sum $V_h$ from \eqref{eq:color-estimator-sketch} and the conditional overlap under partial coloring \eqref{eq:conditional-overlap}. We do so by dynamic programming that avoids naive explicit enumeration of all terms.

\Cref{prop:colorful-coefficients} provides a dynamic programming procedure for evaluating a general class of sums that include both \eqref{eq:color-estimator-sketch} and \eqref{eq:conditional-overlap} as special cases.

\begin{proposition}\label{prop:colorful-coefficients}
Let $n \ge K, b\ge1$. Fix a full coloring $h:[n]\to[K]$, a set
$U\subseteq[n]$, and a partial coloring $\chi:U\to[K]$.
Let $\alpha : [n] \times [n] \times 2^{[K]} \to \mathbb{R}$ be any supplied coefficients. Then, for arbitrary $t=(t_{e})_{e \in E}$, the polynomial
\begin{equation}\label{eq:colorful-alpha-polynomial}
 P(t)\coloneqq
 \sum_{\substack{F\in\mathcal F_{K,b}\\F\text{ colorful under }h}}
 \alpha\bigl(|V(F)|,|V(F)\cap U|,\chi(V(F)\cap U)\bigr)\prod_{e \in F}t_e
\end{equation}
can be evaluated exactly using at most 
\[
 \poly(n)e^{O(K)}\sum_{j=0}^{b-1}\binom nj
\]
arithmetic operations.
\end{proposition}
\begin{proof}[Proof of \Cref{prop:colorful-coefficients}]
Fix $h,U,\chi$ and the coefficients $\alpha$. Recall that each
graph $F\in\mathcal F_{K,b}$ has at most $K$ edges, and its
branching set $ B(F)\coloneqq\{v\in V(F):\deg_F(v)\ge4\}$ has size less than $b$.

First, we can enumerate all possible sets $B(F)$. There are at
most $\sum_{r=0}^{b-1}\binom nr$ such sets. Therefore, it is
enough to show that, for each fixed set $B$, we can evaluate
\[
 P_B(t)\coloneqq
 \sum_{\substack{F\in\mathcal F_{K,b}\\
                 F\text{ colorful under }h\\ B(F)=B}}
 \alpha\bigl(|V(F)|,|V(F)\cap U|,\chi(V(F)\cap U)\bigr)
 \prod_{e\in F}t_e
\]
using $\poly(n)e^{O(K)}$ arithmetic operations. For the rest of the proof, fix $B$ and write $q=|B|$.
We may assume that $h$ is injective on $B$ and that $q\le K/2$, since otherwise $P_B(t)=0$. Indeed, every vertex in $B$ has degree at least four, whereas the sum of all degrees is at most $2K$.

We call the vertices in $B$ \emph{terminal vertices}. Consider a
graph $F$ appearing in the sum above. Every non-terminal vertex
of $F$ has degree exactly two. Thus, starting from a terminal,
we can follow an edge through non-terminal vertices until we
reach another terminal, possibly the one we started from.
This partitions the edge set of $F$ into the following pieces:
\begin{enumerate}
 \item direct edges between terminal vertices;
 \item paths between distinct terminals, with at least one non-terminal vertex;
 \item paths returning to their starting terminal;
 \item disjoint cycles containing no terminal vertices.
\end{enumerate}
Since $F$ is colorful, its non-terminal
vertices have distinct colors, all outside $h(B)$, so they must belong to the set
\[
 I\coloneqq\{i\in[n]\setminus B:h(i)\notin h(B)\}.
\]

Now computing $P_B(t)$ amounts to summing the weights of all valid
arrangements of these pieces. We will first group together graphs $F$ that share the following data:
\[
 \begin{aligned}
 e&=|F|, &
 A&=h(V(F)\setminus B), &
 \boldsymbol d&=(\deg_F(a))_{a\in B},\\
 j&=|(V(F)\setminus B)\cap U|, &
 C&=\chi((V(F)\setminus B)\cap U).
 \end{aligned}
\]
These data determine all three arguments of $\alpha$, so it suffices to compute the sum of the edge weights $\prod_{e\in F}t_e$ over all $F$ that share these data.
We apply the common coefficient $\alpha$ only after this sum has
been computed. This is useful because $\alpha$ could be  arbitrary.

To carry out this grouping while assembling the pieces, we use commuting formal variables
\[
 x,\qquad (y_c)_{c\in[K]\setminus h(B)},\qquad
 (v_a)_{a\in B},\qquad \tau,\qquad (z_c)_{c\in[K]}.
\]
We represent the data $(e,A,\boldsymbol d,j,C)$ by the formal monomial
\[
 x^e y_Av^{\boldsymbol d}\tau^jz_C,
 \qquad \text{ where } \;\;
 y_A\coloneqq\prod_{c\in A}y_c,\quad
 v^{\boldsymbol d}\coloneqq\prod_{a\in B}v_a^{d_a},\quad
 z_C\coloneqq\prod_{c\in C}z_c.
\]

The variables are formal and never assigned any numerical values. We store their coefficients, so that the resulting formal polynomial is simply a convenient way to represent a whole table of weighted sums as one algebraic object. Multiplying monomials will then combine the data of two pieces and discard incompatible choices.

Concretely, we impose the relations $ y_c^2=0$ and $z_c^2 = z_c$ because $F$ must be colorful and $C$ should record the set of repetitions.  We also only retain the monomials satisfying
\[
 0\le e,j\le K,\qquad
 A\subseteq[K]\setminus h(B),\qquad C\subseteq[K], \qquad 
 \boldsymbol d\in\mathbb Z_{\ge0}^{B},\qquad
 |\boldsymbol d|\coloneqq\sum_{a\in B}d_a\le2K.
\]
Every graph of interest satisfies these bounds, so the number of possible data combinations (or monomials) are at most $e^{O(K)}.$

We now construct the weighted sum of individual pieces. For $i\in I$,  we put
\[
 \omega_i\coloneqq
 \begin{cases}
  \tau z_{\chi(i)},&i\in U,\\
  1,&i\notin U.
 \end{cases}
\]
Let $W$ and $D$ be $I \times I$ matrices defined by 
\[
 W_{ij}\coloneqq
 \begin{cases}xt_{ij},&i\ne j,\\0,&i=j,\end{cases}
 \qquad \quad 
 D_{ij}\coloneqq
 \begin{cases}0 ,&i\ne j,\\y_{h(i)}\omega_i,&i=j,\end{cases}
\]
And finally, define $(r_a)_{a \in B}$ to be $I$ indexed vector such that 
\[
 (r_a)_i\coloneqq xt_{ai},\quad \text{ and } \quad 
 r(v)\coloneqq\sum_{a\in B}v_a r_a.
\]
With matrices $W$ and $D$, and vectors $(r_a)_{a \in B}$,  we can now conveniently count weighted paths and cycles among interior vertices. For example, for $\ell\ge1$ we have
\[
 \begin{aligned}
 r_a^{\mathsf T}(DW)^{\ell-1}Dr_{a'}
 ={}&\sum_{\substack{i_1,\ldots,i_\ell\in I\\
                     i_s\ne i_{s+1}\ (1\le s<\ell)}}
 x^{\ell+1}t_{ai_1}
 \left(\prod_{s=1}^{\ell-1}t_{i_si_{s+1}}\right)t_{i_\ell a'}
 \prod_{s=1}^{\ell}y_{h(i_s)}\omega_{i_s}.
 \end{aligned}
\]
Each term follows a walk from $a \in B$ to $a' \in B$ with $\ell$ interior vertices. The relations $y_c^2=0$ remove every walk whose interior
colors repeat. Consequently, the surviving terms are precisely the distinct colorful paths from $a$ to $a'$ that use $\ell$ internal vertices, and this sum also includes returning paths when $a=a'$ and $\ell\ge2$.

We now sum over all possible choices of a single piece, leaving
direct edges between terminals for later. Each piece contributes
its edge weight together with the formal variables recording its
data. Counting each piece once, we obtain the polynomial
\begin{equation}\label{eq:piece-polynomial}
 \mathcal H_B\coloneqq{}
 \sum_{\ell=3}^{K}\frac{\Tr((DW)^\ell)}{2\ell}
 +\sum_{\substack{a,a'\in B\\a<a'}}
       v_av_{a'}r_a^{\mathsf T}Dr_{a'}+\frac12r(v)^{\mathsf T}
       \left(\sum_{\ell=2}^{K}(DW)^{\ell-1}D\right)r(v).
\end{equation}
Let us explain the three terms. The first term describes all simple cycles with $\ell$ interior vertices. Every undirected cycle appears $2\ell$ times in the trace, once for each starting vertex and traversal direction, so we divide by $2\ell$ to fix overcounting.

The second term counts paths $a - i - a'$ with exactly one interior vertex $i$. Their endpoints must be distinct, and $a<a'$ counts each path once.
The last term counts paths with at least two interior vertices.
Expanding the two copies of $r(v)$ chooses their terminal endpoints.
Each path appears twice, in each orientation, so the factor $1/2$ counts every such piece once. Terms with more than $K$ edges are discarded by our
truncation rule.

Next, we want to assemble these pieces in all possible ways. Accordingly, define
\begin{equation}\label{eq:piece-assembly}
 \mathcal G_B\coloneqq
 \prod_{\substack{a,a'\in B\\a<a'}}(1+xt_{aa'}v_av_{a'})
 \sum_{m=0}^{K}\frac{\mathcal H_B^m}{m!}.
\end{equation}
The product $\prod_{a<a'}(1+xt_{aa'}v_av_{a'})$ sums over all
choices of direct edges within  $B$. On the other hand, each
term in $\mathcal H_B^m$ chooses an ordered list of $m$ pieces with at least one interior vertex. The relations $y_c^2=0$ ensure that surviving pieces have disjoint
interior color sets and are therefore distinct. Each collection
appears in exactly $m!$ orders, so dividing by $m!$ counts it once.

Thus, $\mathcal G_B$ records the weighted sums for all possible
arrangements of pieces. Since the decomposition into pieces is unique, for any retained data $(e,A,\boldsymbol d,j,C)$, the coefficient
\[
 [x^e y_Av^{\boldsymbol d}\tau^jz_C]\mathcal G_B
\]
is the sum of the edge weights of all colorful graphs with
branching set $B$ and the given data. Multiplying this coefficient by
\[
 \alpha\bigl(|B|+|A|,\ |B\cap U|+j,\ \chi(B\cap U)\cup C\bigr)
\]
and summing over the retained terms gives $P_B(t)$.

Finally, each polynomial involved has at most $M=e^{O(K)}$
coefficients. Addition takes at most $M$ arithmetic operations,
and multiplication takes at most $2M^2$. The matrices have order
at most $n$, and the expressions above require only polynomially
many table additions and multiplications in $n$ and $K$.
Hence, computing $\mathcal G_B$ and applying the final weights
takes $\poly(n)e^{O(K)}$ arithmetic operations. Summing over
the possible choices of $B$ gives the claimed bound.
\end{proof}

\begin{proof}[Proof of \Cref{lem:colorful-evaluator}]
Apply \Cref{prop:colorful-coefficients} with $U=\varnothing$,
the empty map $\chi$, and $\alpha(v,m,C)=1/p_v$.
Then \eqref{eq:colorful-alpha-polynomial} is exactly $V_h$.
The probabilities $p_v$ are positive and computable in $\poly(n,K)$ operations.
\end{proof}

\appendix

\section{Extension to mixed p-spin models}\label{app:mixed}

\subsection{The ALR expansion}\label{app:mixed-expansion}

In this section we give the high-temperature expansion of Aizenman, Lebowitz, and Ruelle~\cite{ALR87,ALR88} for the mixed $p$-spin model in \eqref{eq:model}. We first group together the terms of the
Hamiltonian that give the same spin monomial. This also accounts for
the repeated indices in \eqref{eq:model}.

For an ordered tuple $\boldsymbol i=(i_1,\ldots,i_p)$, let
$A(\boldsymbol i)\subseteq[n]$ be the set of indices occurring an odd
number of times. Note that $\sigma_{i_1}\dots \sigma_{i_p} = \sigma_{A(\boldsymbol i)}$, where for a set $A$ we denote $\sigma_A\coloneqq\prod_{i\in A}\sigma_i$. For $A\subseteq[n]$, define
\begin{equation}\label{eq:mixed-couplings}
 J_A\coloneqq\sum_{p=2}^{p_{\max}}
 \frac{\gamma_p}{n^{(p-1)/2}}
 \sum_{\substack{\boldsymbol i = (i_1,\dots,i_p)\in[n]^p\\A(\boldsymbol i)=A}}
 G_{i_1,\ldots,i_p}
 \qquad\text{and}\qquad
 v_A\coloneqq\EE_G J_A^2\,.
\end{equation}
The couplings $J_A$ are independent centered Gaussians, since the sums defining
them use disjoint entries of $G$. Let $E$ be the collection of nonempty interaction sets with positive variance:
\[
 E\coloneqq\{A\subseteq[n]:1\le |A|\le p_{\max},\ v_A>0\}\,.
\]
Then
\begin{equation}\label{eq:mixed-hamiltonian}
 H_G(\sigma)=J_\varnothing+\sum_{A\in E}J_A\sigma_A
\end{equation}

For real $\beta$, let $t_A(\beta)\coloneqq\tanh(\beta J_A)$. Similarly to the SK model case, for every hyperedge $A \in E$, 
\[
 e^{\beta J_A\sigma_A}
 =\cosh(\beta J_A)(1+t_A(\beta)\sigma_A)\,,
\]
so we can rewrite the normalized partition function 
$$
Y_G(\beta) = \frac{Z_G(\beta)} {\EE_G Z_G(\beta)} = \frac{e^{\beta J_\varnothing}\prod_{A \in E}\cosh(\beta J_A)}{\EE_G Z_G(\beta)} \cdot \underbrace{\EE_\sigma\prod_{A\in E}(1+t_{A}(\beta)\sigma_A)}_{\eqcolon P(t)}\,.
$$
The prefactor is positive for real $\beta$ and can be computed in $O(n^{p_{\max}})$ operations: indeed, Gaussian integration gives $\EE_G Z_G(z)=e^{n\xi(1)z^2/2}$. This reduces the task of approximating $Y_G(\beta)$ to approximating $P(t(\beta))$.
For a subset of hyperedges $F\subseteq E$, define $t_F\coloneqq\prod_{A\in F}t_A$, with $t_\varnothing\coloneqq1$, and
\[
 \deg_F(i)\coloneqq|\{A\in F:i\in A\}|,
 \qquad V(F)\coloneqq\{i:\deg_F(i)>0\}.
\]
We call $F$ \emph{even} if every $\deg_F(i)$ is even. The same identity used in \eqref{eq:ov-graph-expansion} now gives
\begin{equation}\label{eq:mixed-alr}
 P(t) = \EE_\sigma\prod_{A\in E}(1+t_{A}\sigma_A) = \sum_{\substack{F\subseteq E\\F\text{ even}}}t_F\,.
\end{equation}
Fix constants $\lambda, R_0$, such that
\[
 \beta_*<R_0<\sqrt\lambda<\beta_{\mathrm{2nd}}.
\]

All constants in this appendix may depend on the mixture $\xi$ and
these fixed parameters. We use the coefficient norm in \Cref{def:w_norm}, with
an interaction-dependent weight in place of $w=\lambda/n$.
\begin{definition}
    Let $w_A\coloneqq\lambda v_A$ and $w_F\coloneqq\prod_{A\in F}w_A$. For a multilinear polynomial $f(t) = \sum_{F \subseteq E} f_F t_F$, define 
    \begin{equation}\label{eq:mixed-coefficient-norm}
 \|f\|_w^2\coloneqq\sum_{F\subseteq E}|f_F|^2w_F,\qquad
 \langle f,h\rangle_w\coloneqq\sum_{F\subseteq E}f_F\overline{h_F}w_F\,.
\end{equation}
\end{definition}
The orthogonality argument in \Cref{sec:analysis} applies without change: if
$F\ne F'$, an interaction in their symmetric difference contributes
an odd function of an independent centered Gaussian. Hence, for
real $|\beta|\le\beta_*$,
\begin{equation}\label{eq:mixed-real-error}
 \EE_G|f(t(\beta))|^2
 =\sum_F|f_F|^2\prod_{A\in F}\EE_G\tanh^2(\beta J_A)
 \le\sum_F|f_F|^2\prod_{A\in F}\beta^2v_A
 \le\|f\|_w^2.
\end{equation}
In particular, for any deterministic multilinear polynomial $\widehat P$, 
$$
\EE_G |\widehat P(t(\beta)) - P(t(\beta))|^2 \leq \|\widehat P - P\|_w^2\,.
$$

\subsection{Truncation bounds}\label{app:mixed-truncation}

For a subhypergraph $F \subseteq E$, as in \Cref{sec:truncation}, write
\[
 B(F)\coloneqq\{i:\deg_F(i)\ge4\},\qquad s(F)\coloneqq|B(F)|.
\]
Additionally, we keep track of the number of interactions that are not pairs:
\[
 r(F)\coloneqq|\{A\in F:|A|\ne2\}|.
\]
Define the retained family and its polynomial by
\begin{equation}\label{eq:mixed-retained-family}
 \mathcal F_{K,b}\coloneqq
 \{F\subseteq E:F\text{ even},\ |F|\le K,\ s(F)<b,\ r(F)<b\},
 \qquad
 P_{K,b}(t)\coloneqq\sum_{F\in\mathcal F_{K,b}}t_F.
\end{equation}

\begin{lemma}\label{lem:mixed-truncation}
There are constants $\kappa>0$ and $n_0$ such that, for $n\ge n_0$
and $\log n\le L\le\kappa n$, there are integer cutoffs $K\ge0$
and $b\ge1$ with $K=O(L)$ and $b=O\left(\frac{L}{\log(n/L)}\right)$, satisfying
\[
 \|P-P_{K,b}\|_w^2\le e^{-L}\,.
\]
All constants depend only on the fixed mixture and $\lambda$.
\end{lemma}

\begin{proof}
We follow the proof of \Cref{lem:sk-truncation}, with one additional
tail for $r(F)$. As in \eqref{eq:truncation_ub_two_terms},
\begin{equation}\label{eq:mixed-three-tails}
 \|P-P_{K,b}\|_w^2
 \le\sum_{\substack{F\text{ even}\\|F|>K}}w_F
   +\sum_{\substack{F\text{ even}\\s(F)\ge b}}w_F
   +\sum_{\substack{F\text{ even}\\|F|\le K,\ r(F)\ge b}}w_F.
\end{equation}

First, we bound the total weighted mass, $\|P\|_w^2$, replacing the SK overlap
bound in \eqref{eq:graph-mass} by its mixed-model analog.
For independent uniform signs $\rho_1,\ldots,\rho_n$, let
$S_n\coloneqq\sum_{i=1}^n\rho_i$. \begin{equation}\label{eq:mixed-mass}
 \begin{aligned}
 \|P\|_w^2
 =\EE_\rho\prod_{A\in E}(1+w_A\rho_A)\le\EE_\rho\exp\left\{\sum_{A\in E}w_A\rho_A\right\}
 =e^{-\lambda v_\varnothing}
       \EE_\rho\exp\{n\lambda\xi(S_n/n)\}\,,
 \end{aligned}
\end{equation}
where we denote $\rho_A\coloneqq\prod_{i\in A}\rho_i$ and use $\sum_{A\in E}v_A\rho_A=n\xi(S_n/n)-v_\varnothing$. The second-moment estimate of
\cite[Proposition~3.4]{BHL25} gives, for every fixed
$0\le u<\beta_{\mathrm{2nd}}^2$,
\begin{equation}\label{eq:mixed-overlap-bound}
 \sup_{m\ge1}\EE_\rho\exp\{mu\xi(S_m/m)\}\le C_u.
\end{equation}
The proposition gives the bound for all sufficiently large $m$;
the remaining finitely many values can be included by increasing
$C_u$. By symmetry, the same bound, with an additional factor of
two, holds with $|S_m|$ in place of $S_m$, which gives
\begin{equation}\label{eq:mixed-mass-final}
\|P\|_w^2 \leq C
\end{equation}
for some constant $C$ depending only on $\xi$ and $\lambda$.

Choose $\theta>1$ such that $\theta\lambda<\beta_{\mathrm{2nd}}^2$.
Applying \eqref{eq:mixed-mass-final} with weights $\theta w_A$, exactly as
in \eqref{eq:graph-edge-tail}, yields
\begin{equation}\label{eq:mixed-edge-tail}
 \sum_{\substack{F\text{ even}\\|F|>K}}w_F
 \le\theta^{-K}\sum_{F\text{ even}}\theta^{|F|}w_F
 \le Ce^{-cK}\,.
\end{equation}

To bound the second term, fix a particular set
$S\subseteq[n]$ of $b$ vertices, as in the proof of
\eqref{eq:graph-branching_tail}. Set
\[
 M_S\coloneqq\sum_{\substack{F\text{ even}\\
                       \deg_F(i)\ge4\ \forall i\in S}}w_F\,,
\]
so that $\sum_{\substack{F\text{ even}\\s(F)\ge b}}w_F \leq \sum_{S \in \binom{[n]}{b}} M_S$. Let $\theta > 1$ be chosen below and let $a_i\coloneqq\theta^{\ind\{i\in S\}}$. Marking each incidence at a
vertex of $S$ gives
\begin{equation}\label{eq:mixed-marked-mass}
 \begin{aligned}
 \theta^{4b}M_S \le\sum_{F\text{ even}}w_F
             \theta^{\sum_{i\in S}\deg_F(i)} \le\EE_\rho\exp\left\{\sum_{A\in E}
                     w_A\rho_A\prod_{i\in A}a_i\right\}
 \le\EE_\rho\exp\left\{n\lambda\xi\left(
                         \frac1n\sum_i a_i\rho_i\right)\right\}\,,
 \end{aligned}
\end{equation}
where the last inequality follows from $a_i \ge 1$.
Let $U\coloneqq\sum_{i\notin S}\rho_i$ and
$V\coloneqq\sum_{i\in S}\rho_i$, so that 
$\sum_i a_i \rho_i = U + \theta V$. Choose $\eta>0$ with
$(1+\eta)\lambda<\beta_{\mathrm{2nd}}^2$. Independence and convexity give
\begin{equation}\label{eq:mixed-marked-split}
 \EE_\rho e^{n \lambda \xi((U+\theta V)/n)}
 \le \EE_\rho e^{n \lambda \xi(|U+\theta V|/n)} \le \EE_\rho e^{n\lambda(1+\eta)\xi(|U|/n)}
       \EE_\rho e^{n\lambda c_\eta \xi(\theta|V|/n)}\,,
\end{equation}
where $c_\eta > 0$ is some constant. \eqref{eq:mixed-overlap-bound} upper bounds the first factor by taking $m=n-b$.
For the second factor, note that $\xi(q) \leq \xi(1) q^2$ for $q\in [0,1]$, so the Gaussian
integration argument in the proof of \Cref{fact:overlap-moment}
yields
\[
 \EE_\rho e^{n\lambda c_{\eta}\xi(\theta|V|/n)}
 \le\EE_\rho\exp\{\lambda c_{\eta}\xi(1)\theta^2 V^2/n\}
 \le\left(1-2\lambda c_\eta \xi(1)\theta^2b/n\right)^{-1/2}\,,
\]
provided $\theta^2 < \frac n {2 b \lambda c_\eta \xi(1)}$ and $\theta |V| / n \leq 1$. Since $b \leq \kappa_0 n$, for sufficiently small constants $\kappa_0$ and $c_0$, we can set
\[
 \theta^2\coloneqq1+c_0n/b\,,
\]
satisfying $2\lambda c_\eta \xi(1)\theta^2b/n \leq 2\lambda c_\eta \xi(1)(\kappa_0 + c_0) < 1$ and $\theta |V| / n \leq \kappa_0 + c_0 \leq 1$. This upper bounds the second factor by a constant and gives 
$
M_S \leq C \theta^{-4b}\,.
$
Summing over the choices of $S$, proves, similarly to \Cref{lem:sk-truncation},
\begin{equation}\label{eq:mixed-branching-tail}
 \sum_{\substack{F\text{ even}\\s(F)\ge b}}w_F
 \le C\binom nb\theta^{-4b}
 \le C(Cb/n)^b\,.
\end{equation}

It remains to bound the third term. Fix an integer $h\ge1$. Again letting $\rho_1, \dots, \rho_n$ to be independent signs and $S_n = \sum_i \rho_i$, we have 
\begin{equation}\label{eq:term3_exp}
\sum_{\substack{F\text{ even}\\r(F)=h}}w_F \le \frac 1 {h!} \EE_\rho \left[ \exp\left( \sum_{|A| = 2} w_A \rho_A \right)\left( \sum_{|A| \neq 2} w_A \rho_A \right)^h\right] 
\end{equation}
We now upper bound the weights $w_A$. For a fixed $A \subseteq [n]$ and $p \geq |A|$, there are $O(n^{(p-|A|)/2})$ $p$-tuples $\boldsymbol i$ for which $A(\boldsymbol i) = A$. Then, the variances $v_A$ depend only on $|A|$ and satisfy
\begin{equation}\label{eq:mixed-variance-bounds}
 v_\varnothing=O(1),\qquad
 v_{\{i\}}=O(n^{-1}),\qquad
 v_{\{i,j\}}=\frac{2\gamma_2^2}{n}+O(n^{-2}),\qquad
 v_A=O(n^{1-|A|})\quad(|A|\ge3).
\end{equation}
Since $I(q)=q^2/2+O(q^4)$, \eqref{eq:second-moment-threshold}
and $\lambda<\beta_{\mathrm{2nd}}^2$ imply
$2\lambda\gamma_2^2<1$. Consequently, there is a fixed
$0<a<1/2$ such that, for all sufficiently large $n$,
\[
 w_{\{i,j\}}\le\frac{2a}{n},\qquad
 w_{\{i\}}\le\frac Cn,\qquad
 w_A\le Cn^{1-|A|}\quad(|A|\ge3).
\]
With these bounds we further upper bound \eqref{eq:term3_exp} as\footnote{After expanding products and exponentials, every resulting spin monomial has
expectation zero or one, so increasing nonnegative coefficients
can only increase the expectation.}
\[
 \begin{aligned}
 \sum_{\substack{F\text{ even}\\r(F)=h}}w_F
 &\le\frac{C^h}{h!}\EE_\rho\left[
 e^{aS_n^2/n}
 \left(\frac{S_n}{n}
       +\sum_{j=3}^{p_{\max}}\frac{S_n^j}{n^{j-1}}\right)^h\right]\\
 &\le\frac{C^h}{h!}\EE_\rho\left[
 e^{aS_n^2/n}
 \left(\frac{|S_n|}{n}+\frac{|S_n|^3}{n^2}\right)^h\right].
 \end{aligned}
\]
For a fixed $a<a'<1/2$, the Gaussian integration bound from the
proof of \Cref{fact:overlap-moment} gives
$\EE_\rho e^{a'S_n^2/n}\le(1-2a')^{-1/2}$. Hence, for every
integer $m\ge1$,
\[
 \EE_\rho\bigl[|S_n|^m e^{aS_n^2/n}\bigr]
 \le (1-2a')^{-1/2}
       \sup_{x\ge0}x^m e^{-(a'-a)x^2/n}
 \le C(Cnm)^{m/2}.
\]
Apply this with $m=h,3h$ and use
$(x+y)^h\le2^{h-1}(x^h+y^h)$ for $x,y\ge0$ and
$h!\ge(h/e)^h$ to obtain
\[
 \sum_{\substack{F\text{ even}\\r(F)=h}}w_F
 \le C\left[
       \left(\frac{C}{nh}\right)^{h/2}
       +\left(\frac{Ch}{n}\right)^{h/2}\right]
 \le C(Ch/n)^{h/2}.
\]
Decreasing $\kappa_0$ so that $CK/n\le1/4$ whenever
$K\le\kappa_0n$, we can sum a geometric series to get
\begin{equation}\label{eq:mixed-nonpair-tail}
 \sum_{\substack{F\text{ even}\\|F|\le K,\ r(F)\ge b}}w_F
 \le C\sum_{h=b}^K(CK/n)^{h/2}
 \le C(CK/n)^{b/2}.
\end{equation}

Plugging the term-wise bounds
\eqref{eq:mixed-edge-tail}, \eqref{eq:mixed-branching-tail}, and
\eqref{eq:mixed-nonpair-tail} into  \eqref{eq:mixed-three-tails}, we obtain
\[
 \|P-P_{K,b}\|_w^2
 \le Ce^{-cK}+C(Cb/n)^b+C(CK/n)^{b/2}.
\]
Choose sufficiently large constants $A,A'$ and set
\[
 K\coloneqq\lceil AL\rceil,\qquad
 b\coloneqq\left\lceil\frac{A'L}{\log(n/L)}\right\rceil.
\]
For sufficiently small $\kappa>0$, the assumption
$\log n\le L\le\kappa n$ ensures $K,b\le\kappa_0n$ and
\[
 \log\frac{n}{CK}\ge\frac12\log\frac nL,\qquad
 \log\frac{n}{Cb}\ge\frac12\log\frac nL.
\]
Thus the three terms are at most $Ce^{-cAL}$,
$Ce^{-A'L/2}$, and $Ce^{-A'L/4}$, respectively. Taking $A,A'$
sufficiently large makes each at most $e^{-L}/3$.
Since $L/\log(n/L)\ge1$, the cutoffs have the asserted orders.
This proves the lemma.
\end{proof}

\subsection{Deterministic estimator via color coding}
\label{app:mixed-color-coding}

In this section, we extend the
color-coding construction of \Cref{sec:estimator} to a more general setting. Write $\mathcal F\coloneqq\mathcal F_{K,b}$ and
$P_{\mathcal F}\coloneqq P_{K,b}$ for the retained family and its
polynomial from \eqref{eq:mixed-retained-family}. 
The key observation is that every non-pair interaction lies in the
small set
\[
 B_*(F): =B(F)\cup\bigcup_{{A\in F, \; |A|\ne2}}A.
\]
Outside this set, every supported vertex has degree two and is involved only in pair interactions. Thus, after fixing $B_*(F)$, we can reuse the dynamic program developed in \Cref{subsec:color-dp} and incorporate the non-pair interactions within the terminal set. Our ultimate goal is to prove the following extension of \Cref{prop:sk-approximation}.

\begin{proposition}\label{prop:mixed-approximation}
Let $n\ge n_0$ and $\log n\le L\le\kappa n$, and choose $K,b$ as
in \Cref{lem:mixed-truncation}. There is a deterministic multilinear
polynomial $\widehat P_{\mathcal F}$, supported on $\mathcal F$,
such that
\[
 \|\widehat P_{\mathcal F}-P_{\mathcal F}\|_w^2\le e^{-L}.
\]
Its arithmetic circuit can be constructed and evaluated in
$\poly(n)e^{O(L)}$ operations and is independent of the disorder
and temperature.
\end{proposition}

We start by setting up the color coding. Since every supported vertex of an even hypergraph has degree at
least two, we have
\[
 2|V(F)|\le\sum_i\deg_F(i)
 =\sum_{A\in F}|A|\le p_{\max}K.
\]
We therefore use $Q\coloneqq\lfloor p_{\max}K/2\rfloor$ colors. Fix color-class sizes
$n_1,\ldots,n_Q$ differing by at most one and summing to $n$.
All balanced colorings below have these prescribed sizes.

Let $p_v$ be the probability that a fixed set of $v$ vertices is
colorful under a uniformly random balanced coloring, with $p_0=1$.
For $h:[n]\to[Q]$, define
\begin{equation}\label{eq:mixed-color-estimator}
 V_h(t)\coloneqq\sum_{F\in\mathcal F}
 \frac{\ind\{h\text{ is injective on }V(F)\}}{p_{|V(F)|}}t_F.
\end{equation}
Then $\EE_hV_h=P_{\mathcal F}$.

As in the pure SK case, we need an evaluator that also allows
coefficients depending on a partial coloring. This will let us
both evaluate $V_h$ and compute the conditional expectations used
to choose the colorings deterministically. 

\begin{proposition}\label{prop:mixed-colorful-coefficients}
Fix integers $K,b\ge1$ and $n\ge Q$, a coloring $h:[n]\to[Q]$,
a set $U\subseteq[n]$, a partial coloring $\chi:U\to[Q]$, and
coefficients $\alpha:[n]^2\times2^{[Q]}\to\mathbb R$. Then the polynomial
\begin{equation}\label{eq:mixed-colorful-alpha}
 \sum_{\substack{F\in\mathcal F\\F\text{ colorful under }h}}
 \alpha\bigl(|V(F)|,|V(F)\cap U|,\chi(V(F)\cap U)\bigr)t_F
\end{equation}
can be evaluated exactly for arbitrary interaction inputs $t=(t_A)_{A\in E}$ using $\poly(n)e^{O(K)}n^{q_*}$ arithmetic operations.
\end{proposition}

\begin{proof}
Fix $h, U, \chi$ and the coefficients $\alpha$. For every $F \in \mathcal{F}$, we have $|B_*(F)| \le b + b p_{\max} =: q_*$. Thus, we can enumerate all possible sets $B = B_*(F)$, and it suffices to show that, for each fixed $B$, we can evaluate
\begin{equation*}
    \sum_{\substack{F \in \mathcal F \\ F \text{ colorful under } h \\ B_*(F) = B }}  \alpha\bigl(|V(F)|,|V(F)\cap U|,\chi(V(F)\cap U)\bigr)t_F
\end{equation*}
using $\poly(n) e^{O(K)}$ operations. We now fix $B$ and write $q = |B|$. We assume that $h$ is injective on $B$; otherwise, the sum is empty. The key observation is that all non-pair interactions are contained in $B$, while every supported vertex outside $B$ has
degree two and is incident only to pairs. For this reason, most of the work has already been done in the proof of \Cref{prop:colorful-coefficients}, and the only remaining step is to deal with non-uniform pair interactions within $B$.

We now adopt the notation used in the proof of \Cref{prop:colorful-coefficients}. First, construct the table $\mathcal G_B$ from \eqref{eq:piece-assembly}, using $Q$ colors and inputs
$t_{ij}=t_{\{i,j\}}$ (with $t_{ij}=0$ for absent pairs). To keep track of the non-pair interactions, introduce a formal variable
$u$ to record the number of non-pair interactions within $B$, and formal variables $\eta_a$ for $a \in B$, with
$\eta_a^2=\eta_a$, to record the terminal vertices involved in non-pair interactions. Set
\begin{equation}\label{eq:mixed-piece-assembly}
 \mathcal G_B^{\mathrm{mix}}\coloneqq
 \mathcal G_B
 \prod_{\substack{A\in E\\A\subseteq B,\ |A|\ne2}}
 \left(1+xu\,t_A\prod_{a\in A}v_a\eta_a\right).
\end{equation}
The table given by $\mathcal{G}_B^{\mathrm{mix}}$ is indexed by monomials
\[
 x^e u^r y_Sv^{\boldsymbol d}\tau^jz_C\eta_M.
\]
Here $e$ is the total number of interactions, $r$ is the number
of non-pair interactions, $S \subset [Q] \setminus h(B)$ is the set of $h$-colors of the non-terminal vertices, and $\boldsymbol d$ is the vector of total terminal degrees. The indices $j$ and $C$ record, respectively, the number of non-terminal vertices in $U$ and their set of $\chi$-colors; $M\subseteq B$ records the terminals met by
non-pair interactions. As before, the relations $y_c^2=0$ enforce
colorfulness, while the relations $z_c^2=z_c$ allow us to record sets of partial colors.

Finally, these data fully determine the value of $\alpha$ and whether the configuration is valid. For each fixed $B$, this augmented table has at most
\[
 \poly(K)\,2^{2Q+q}\binom{p_{\max}K+q}{q}=e^{O(K)}
\]
entries, and its construction still takes $\poly(n)e^{O(K)}$ operations. Summing over the possible sets $B$ proves the stated bound.
\end{proof}

We now use this evaluator to extend the derandomization procedure from \Cref{subsec:color-derandomization} to the mixed model here. The following lemma gives the colorings needed to prove
\Cref{prop:mixed-approximation}.

\begin{lemma}
\label{lem:mixed-color-overlap}
Under the setup above, let $q_*\coloneqq(1+p_{\max})b$.
For every integer $T\ge1$, we can construct
balanced colorings $h_1,\ldots,h_T$ such that, for each $s \in [T]$
\begin{equation}\label{eq:mixed-greedy-overlap}
 \sum_{i<s}\langle V_{h_i},V_{h_s}\rangle_w
 \le(s-1)\|P_{\mathcal F}\|_w^2.
\end{equation}
The construction uses $\poly(n,T)e^{O(K)}n^{q_*}$ arithmetic operations.
\end{lemma}

\begin{proof}
Let $H$ be a uniformly random balanced coloring with the prescribed color-class sizes $n_1,\ldots,n_Q$.  Fix a coloring $h:[n]\to[Q]$ and a partial coloring
$\chi:U\to[Q]$ admitting a balanced completion. For a support
$S\subseteq[n]$ of size at most $Q$, its conditional probability of being colorful depends only on $ v=|S|$, $m=|S\cap U|$, and $C=\chi(S\cap U).$ Define
\begin{equation}\label{eq:mixed-conditional-coefficient}
 \alpha_\chi(v,m,C)\coloneqq
 \frac{\PP[H\text{ is injective on }S\mid H|_U=\chi]}{p_v^2},
\end{equation}
This expectation is explicit, so the whole table $\alpha_\chi$ can be computed in $\poly(n)e^{O(Q)}$ operations. Now note that for the fixed coloring $h$,
\begin{equation}\label{eq:mixed-conditional-overlap}
 \EE_H\!\left[\langle V_h,V_H\rangle_w\mid H|_U=\chi\right]
 =\sum_{\substack{F\in\mathcal F\\F\text{ colorful under }h}}
 \alpha_\chi\bigl(|V(F)|,|V(F)\cap U|,\chi(V(F)\cap U)\bigr)w_F.
\end{equation}
\Cref{prop:mixed-colorful-coefficients}, with $t_A=w_A$, therefore
computes each conditional overlap exactly in $\poly(n)e^{O(K)}n^{q_*}$ time. Suppose now that $h_1,\ldots,h_{s-1}$ have already been selected.
For a partial coloring $\chi:U\to[Q]$ admitting a balanced
completion, define
\[
 \Phi_s(\chi)\coloneqq
 \EE_H\!\left[\sum_{i<s}\langle V_{h_i},V_H\rangle_w
          \mid H|_U=\chi\right].
\]
Since the weights depend only on interaction sizes,
$\langle V_h,P_{\mathcal F}\rangle_w$ is the same for every balanced
coloring $h$; averaging over $h$ shows that it equals
$\|P_{\mathcal F}\|_w^2$. Unbiasedness therefore gives
\[
 \Phi_s(\varnothing)
 =(s-1)\|P_{\mathcal F}\|_w^2.
\]
Assign color to vertices in the order $1,\ldots,n$. At each step, evaluate $\Phi_s$ for every extension by a color and choose an extension that minimizes it. By the law of total expectation, the current value is a weighted average of the values of these extensions, so the chosen
extension does not increase $\Phi_s$. The resulting coloring $h_s$
is balanced and satisfies
\[
 \sum_{i<s}\langle V_{h_i},V_{h_s}\rangle_w
 =\Phi_s(h_s)
 \le\Phi_s(\varnothing)
 =(s-1)\|P_{\mathcal F}\|_w^2,
\]
as required by \eqref{eq:mixed-greedy-overlap}.

At each of the $n$ steps, we consider at most $Q$ colors and
compute at most $s-1$ conditional overlaps per color. Constructing
all $T$ colorings therefore requires at most $nQT^2$ conditional
overlap evaluations. The cost established above gives the stated
running time.
\end{proof}

We are now ready to prove \Cref{prop:mixed-approximation}.

\begin{proof}[Proof of \Cref{prop:mixed-approximation}]
The calculation in the proof of \Cref{prop:sk-approximation},
with $Q$ in place of $K$, gives $ p_v>8^{-Q}.$ Moreover, since $w_A$ depends only on the size $|A|$, for every
balanced coloring $h$ we have
\begin{equation}\label{eq:mixed-color-mass}
 \langle V_h,P_{\mathcal F}\rangle_w=\|P_{\mathcal F}\|_w^2,
 \quad \text{ and } \quad 
 \|V_h\|_w^2\le8^Q\|P_{\mathcal F}\|_w^2.
\end{equation}
Take the deterministic colorings supplied by
\Cref{lem:mixed-color-overlap} and set
\[
 \widehat P_{\mathcal F}\coloneqq\frac1T\sum_{s=1}^T V_{h_s}.
\]
By \eqref{eq:mixed-color-mass} and
\eqref{eq:mixed-greedy-overlap}, we have that 
\[
 \begin{aligned}
 \|\widehat P_{\mathcal F}-P_{\mathcal F}\|_w^2
 &=\frac1{T^2}\sum_{s=1}^T\|V_{h_s}\|_w^2
   +\frac2{T^2}\sum_{s=1}^T\sum_{i<s}
       \langle V_{h_i},V_{h_s}\rangle_w
   -\|P_{\mathcal F}\|_w^2\\
 &\le\frac{8^Q-1}{T}\|P_{\mathcal F}\|_w^2
 \le\frac{C8^Q}{T},
 \end{aligned}
\]
where the last inequality follows from \eqref{eq:mixed-mass}.
Choosing $T=\lceil C8^Qe^L\rceil$ gives the required error bound.
To evaluate each $V_h$, apply
\Cref{prop:mixed-colorful-coefficients} with $U=\varnothing$ and
$\alpha(v,m,C)=1/p_v$.

It remains to bound the running time. Since $Q=O(K)=O(L)$, we
have $\log T=O(L)$. Also, $q_*=O(b)$, and the cutoffs in
\Cref{lem:mixed-truncation} give, after decreasing $\kappa$ if
necessary, $n^{q_*}=e^{O(L)}.$ Thus both construction and evaluation take
$\poly(n)e^{O(L)}$ operations. The colorings depend only on
$\mathcal F$ and the weights $w_A$, so the resulting circuit is
independent of $G$ and $\beta$.
\end{proof}

Combining this approximation with the truncation bound yields a
single statement valid for every accuracy parameter $L\ge\log n$.

\begin{corollary}\label{cor:mixed-coefficient-approximation}
For every $n\ge2$ and $L\ge\log n$, one can deterministically
construct a multilinear polynomial $\widehat P_L$ satisfying
\[
 \|\widehat P_L-P\|_w^2\le e^{-L}.
\]
Its arithmetic circuit is independent of the disorder and
temperature and can be constructed and evaluated in
$\poly(n)e^{O(L)}$ operations.
\end{corollary}
\begin{proof}
If $n\ge n_0$ and $L+\log2\le\kappa n$, apply
\Cref{lem:mixed-truncation,prop:mixed-approximation} with
$L' \ge  L + 1$. The polynomial
$\widehat P_{\mathcal F}-P_{\mathcal F}$ is supported on
$\mathcal F$, whereas $P-P_{\mathcal F}$ is supported outside it.
Their squared coefficient norms therefore add, giving
\[
 \|\widehat P_{\mathcal F}-P\|_w^2
 =\|\widehat P_{\mathcal F}-P_{\mathcal F}\|_w^2
   +\|P-P_{\mathcal F}\|_w^2
 \le2e^{-L'}=e^{-L}.
\]
For the remaining values of $n,L$, use the exact spin average
from \eqref{eq:mixed-alr}:
\[
 \widehat P_L(t)\coloneqq P(t)
 =2^{-n}\sum_{\sigma\in\{\pm1\}^n}
       \prod_{A\in E}(1+t_A\sigma_A).
\]
This takes $O(2^n n^{p_{\max}})$ operations. Since either
$n<n_0$ or $\kappa n<L+1$, this cost is also
$\poly(n)e^{O(L)}$.
\end{proof}

\subsection{Lower bound for the normalized partition function}
\label{app:mixed-lower-bound}
In this section, we will prove analogue of \Cref{lem:lower-bound} in general setting. 

\begin{lemma}\label{lem:mixed-lower-bound}
There are constants $c>0$, $C\ge1$, and a sequence
$\zeta_n=o(1)$ such that, for every $0<\alpha<1/2$,
\begin{equation}\label{eq:mixed-lower-bound}
 \PP_G\!\left[
   \inf_{|z|\le\beta_*}|Y_G(z)|\ge C^{-1}\alpha^c
 \right]\ge1-\zeta_n-\alpha.
\end{equation}
The constants and the sequence depend only on the mixture and
the fixed radii.
\end{lemma}

\begin{proof}
We follow the proof of \Cref{lem:lower-bound}. Recall that 
\cite[Theorem~1.2]{BHL25} gives that 
\begin{equation}\label{eq:mixed-zero-free}
 \PP_G\!\left[Y_G(z)\ne0\text{ for all }|z|<R_0\right]
 \ge1-\zeta_n
\end{equation}
for some $\zeta_n=o(1)$. The Gaussian calculation in the proof of
\Cref{fact:overlap-moment}  gives
\begin{equation}\label{eq:mixed-complex-second-moment}
 \EE_G|Y_G(z)|^2
 =\EE_\rho\exp\{n|z|^2\xi(S_n/n)\},
\end{equation}
where $S_n=\sum_i\rho_i$ for independent uniform signs $\rho_i$.
Thus \eqref{eq:mixed-overlap-bound} bounds this second moment by
$C_\lambda$ on $|z|=\sqrt\lambda$. Applying
\eqref{eq:an-cauchy-sup} as in the proof of
\Cref{lem:upper-bound}, we obtain
\begin{equation}\label{eq:mixed-uniform-upper-moment}
 \EE_G\sup_{|z|\le R_0}|Y_G(z)|^2
 \le(1-R_0/\sqrt\lambda)^{-2}C_\lambda
 \le C_U,
\end{equation}
where $C_U\ge1$ is fixed. Hence, by Markov's inequality,
$\sup_{|z|\le R_0}|Y_G(z)|\le M\coloneqq\sqrt{C_U/\alpha}$
with probability at least $1-\alpha$.

On the intersection with the event in \eqref{eq:mixed-zero-free},
the Harnack argument in the proof of \Cref{lem:lower-bound}, with
$R_0$ in place of $R_1$, gives
\[
 \inf_{|z|\le\beta_*}|Y_G(z)|
 \ge M^{-2\beta_*/(R_0-\beta_*)}
 =C_U^{-c}\alpha^c,
\]
for $c = b \eta_* / (R_0-\beta_*)$,  so the lemma follows with $C=C_U^c$.
\end{proof}

\subsection{Proof of the main theorem}
\label{app:mixed-completion}

We first extend \Cref{lem:entire-lift,lem:cauchy-sup} to the mixed
model. For a deterministic multilinear polynomial
$f(t)=\sum_{F\subseteq E}f_Ft_F$, define, as in \eqref{eq:tr-lift},
\begin{equation}\label{eq:mixed-entire-lift}
 \mathcal L_Gf(z)\coloneqq
 e^{zJ_\varnothing-n\xi(1)z^2/2}
 \sum_{F\subseteq E}f_F
       \prod_{A\in F}\sinh(zJ_A)
       \prod_{A\in E\setminus F}\cosh(zJ_A).
\end{equation}
This function is entire. For any real $\beta$, it is $f(t(\beta))$
times the normalized prefactor from \Cref{app:mixed-expansion}.
In particular, \eqref{eq:mixed-alr} gives $\mathcal L_GP=Y_G$.

\begin{lemma}\label{lem:mixed-entire-lift}
There are constants $C_0,C_{\mathrm{err}}\ge1$ such that every
deterministic multilinear polynomial $f$ satisfies
\begin{equation}\label{eq:mixed-lift-bound}
 \sup_{|z|\le R_0}\EE_G|\mathcal L_Gf(z)|^2
 \le C_0\|f\|_w^2.
\end{equation}
Consequently, for every deterministic multilinear approximation
$\widehat P$,
\begin{equation}\label{eq:mixed-uniform-error}
 \EE_G\sup_{|z|\le\beta_*}
       |\mathcal L_G\widehat P(z)-Y_G(z)|^2
 \le C_{\mathrm{err}}\|\widehat P-P\|_w^2.
\end{equation}
\end{lemma}

\begin{proof}
The calculation in the proof of \Cref{lem:entire-lift}, using
$v_\varnothing+\sum_{A\in E}v_A=n\xi(1)$, gives
\begin{equation}\label{eq:mixed-lift-norm}
 \EE_G|\mathcal L_Gf(z)|^2
 =e^{|z|^2v_\varnothing}
       \prod_{A\in E}\cosh(|z|^2v_A)
       \sum_{F\subseteq E}|f_F|^2
            \prod_{A\in F}\tanh(|z|^2v_A).
\end{equation}
Here the factor $e^{|z|^2v_\varnothing}$ comes from the constant
interaction. By \eqref{eq:mixed-variance-bounds},
$v_\varnothing=O(1)$ and $\sum_{A\in E}v_A^2=O(1)$. Hence, for
$|z|\le R_0$,
\[
 e^{|z|^2v_\varnothing}\prod_{A\in E}\cosh(|z|^2v_A)
 \le\exp\left\{R_0^2v_\varnothing+
             \frac{R_0^4}{2}\sum_{A\in E}v_A^2\right\}
 \le C_0.
\]
Together with
$\tanh(|z|^2v_A)\le R_0^2v_A\le\lambda v_A=w_A$,
this proves \eqref{eq:mixed-lift-bound}. Applying
\eqref{eq:an-cauchy-sup1} to $\mathcal L_G(\widehat P-P)$,
exactly as in the proof of \Cref{lem:cauchy-sup}, proves
\eqref{eq:mixed-uniform-error} with
$C_{\mathrm{err}}=C_0(1-\beta_*/R_0)^{-2}$.
\end{proof}

\begin{proof}[Proof of \Cref{thm:counting-intro}]
Let $c,C$ be the constants in \Cref{lem:mixed-lower-bound},
and put $a\coloneqq C^{-1}(\delta/2)^c$ and  $C_{\mathrm{err}}$ be the
constant in \Cref{lem:mixed-entire-lift}.. Choose $L=\Theta(\log(n/(\varepsilon\delta)))$ so that
\[
 C_{\mathrm{err}}e^{-L}\le\frac{\delta}{2}\varepsilon^2a^2.
\]
Take $\widehat P_L$ from
\Cref{cor:mixed-coefficient-approximation} and set
$\widehat Y_G\coloneqq\mathcal L_G\widehat P_L$.
Then, from \Cref{lem:mixed-lower-bound,lem:mixed-entire-lift} and Markov's
inequality, we can conclude that
\begin{equation}\label{eq:mixed-final-relative}
 \PP_G\!\left[
 \sup_{|\beta|\le\beta_*}
 \left|\frac{\widehat Y_G(\beta)}{Y_G(\beta)}-1\right|
 >\varepsilon
 \right]
 \le\zeta_n+\frac{\delta}{2}
       +\frac{C_{\mathrm{err}}e^{-L}}{\varepsilon^2a^2}
 \le\zeta_n+\delta.
\end{equation}
For real $\beta$, return
\begin{equation}\label{eq:mixed-output}
 \widehat Z_G(\beta)\coloneqq
 e^{\beta J_\varnothing}
 \prod_{A\in E}\cosh(\beta J_A)\,
 \widehat P_L(t(\beta)).
\end{equation}
Note that by \eqref{eq:mixed-alr} and \eqref{eq:mixed-entire-lift}, this has the same relative error as $\widehat Y_G$. Moreover, construction and evaluation of the estimator take
$\poly(n)e^{O(L)}+O(n^{p_{\max}})
=(n/(\varepsilon\delta))^{O(1)}$ operations, which concludes the proof.
\end{proof}

\bibliographystyle{alpha}
\bibliography{references}

\end{document}